\documentclass[journal,twoside,web]{ieeecolor}
\usepackage{generic}
\usepackage{cite}
\usepackage{amsmath,amssymb,amsfonts}
\usepackage{graphicx}
\usepackage{textcomp}

\usepackage{mathrsfs}
\usepackage{mathtools}
\usepackage{bm}
\usepackage{url}
\usepackage{fancybox}
\usepackage{subfigure}
\usepackage{booktabs}
\usepackage{algorithm}
\usepackage{algorithmicx}
\usepackage{algpseudocode}
\algrenewcommand\algorithmicindent{1em}

\newcommand{\figref}[1]{Fig.~\ref{#1}}
\newcommand{\tabref}[1]{Table~\ref{#1}}
\renewcommand{\algref}[1]{Algorithm~\ref{#1}}

\newcommand{\lemref}[1]{Lemma~\ref{#1}}
\newcommand{\thmref}[1]{Theorem~\ref{#1}}

\newcommand{\appref}[1]{Appendix~\ref{#1}}
\newcommand{\secref}[1]{Section~\ref{#1}}

\newcommand{\R}{\mathbb{R}}

\newcommand{\Z}{\mathbb{Z}}
\newcommand{\N}{\mathbb{N}}

\newcommand{\Ecd}{\mathsf{Ecd}}
\newcommand{\Dcd}{\mathsf{Dcd}}
\newcommand{\Qtz}{\mathcal{Q}}

\newcommand{\diag}[1]{\mathop{\mathrm{diag}}\limits\left(#1\right)}
\renewcommand{\vec}{\mathop{\mathrm{vec}}\limits}

\newcommand{\bfzero}{\mathbf{0}}

\newcommand{\rank}{\mathop{\rm rank}\limits}

\newcommand{\KeyGen}{\mathsf{KeyGen}}
\newcommand{\Enc}{\mathsf{Enc}}
\newcommand{\Dec}{\mathsf{Dec}}
\newcommand{\Add}{\mathsf{Add}}
\newcommand{\Mult}{\mathsf{Mult}}
\newcommand{\Rotate}{\mathsf{Rotate}}
\newcommand{\Relin}{\mathsf{Relin}}
\newcommand{\pk}{\mathsf{pk}}
\newcommand{\sk}{\mathsf{sk}}
\newcommand{\rlk}{\mathsf{rlk}}
\newcommand{\ct}{\mathsf{ct}}
\newcommand{\sfx}{\mathsf{x}}
\newcommand{\sfA}{\mathsf{A}}
\newcommand{\sfB}{\mathsf{B}}
\newcommand{\sfC}{\mathsf{C}}

\newcommand{\sfE}{\mathsf{E}}
\newcommand{\sfF}{\mathsf{F}}
\newcommand{\Acl}{\sfA_{\mathrm{cl}}}
\newcommand{\Aclb}{\bar{\sfA}_{\mathrm{cl}}}
\newcommand{\tldd}{\tilde{d}}
\newcommand{\tldK}{\tilde{K}}
\newcommand{\bfa}{\mathbf{a}}

\newcommand{\bfs}{\mathbf{s}}
\newcommand{\bfe}{\mathbf{e}}
\newcommand{\bfu}{\mathbf{u}}
\newcommand{\bfp}{\mathbf{p}}
\newcommand{\bfr}{\mathbf{r}}
\newcommand{\bfm}{\mathbf{m}}
\newcommand{\bfc}{\mathbf{c}}

\newcommand{\round}[1]{\lfloor{#1}\rceil}
\newcommand{\floor}[1]{\lfloor{#1}\rfloor}

\newtheorem{definition}{Definition}

\newtheorem{lemma}{Lemma}
\newtheorem{theorem}{Theorem}

\newtheorem{remark}{Remark}

\makeatletter
\let\MYcaption\@makecaption
\makeatother
\usepackage{subfigure}
\makeatletter
\let\@makecaption\MYcaption
\makeatother

\allowdisplaybreaks[1]

\def\BibTeX{{\rm B\kern-.05em{\sc i\kern-.025em b}\kern-.08em
    T\kern-.1667em\lower.7ex\hbox{E}\kern-.125emX}}
\begin{document}

\title{Input-Output History Feedback Controller for Encrypted Control with Leveled Fully Homomorphic Encryption}

\author{%
    Kaoru Teranishi, \IEEEmembership{Student Member, IEEE},
    Tomonori Sadamoto, \IEEEmembership{Member, IEEE}, \\
    and Kiminao Kogiso, \IEEEmembership{Member, IEEE}
    \thanks{%
        This work was supported by JSPS Grant-in-Aid for JSPS Fellows Grant Number JP21J22442 and JSPS KAKENHI Grant Number JP22H01509.
    }
    \thanks{%
        Kaoru Teranishi, Tomonori Sadamoto, and Kiminao Kogiso are with the Department of Mechanical and Intelligent Systems Engineering, The University of Electro-Communications, Chofu, Tokyo, 182-8585, Japan (e-mail: teranishi@uec.ac.jp, sadamoto@uec.ac.jp, kogiso@uec.ac.jp).
    }
    \thanks{%
        Kaoru Teranishi is also a Research Fellow of Japan Society for the Promotion of Science.
    }
}

\thispagestyle{empty}
\hspace{-4.5mm}
\fbox{
\begin{minipage}{\textwidth-5mm}\scriptsize
© 20XX IEEE.
Personal use of this material is permitted.
Permission from IEEE must be obtained for all other uses, in any current or future media, including reprinting/republishing this material for advertising or promotional purposes, creating new collective works, for resale or redistribution to servers or lists, or reuse of any copyrighted component of this work in other works.
\end{minipage}
}
\newpage
\setcounter{page}{0}

\maketitle
\thispagestyle{empty}

\begin{abstract}
Protecting the parameters, states, and input/output signals of a dynamic controller is essential for securely outsourcing its computation to an untrusted third party.
Although a fully homomorphic encryption scheme allows the evaluation of controller operations with encrypted data, an encrypted dynamic controller with the encryption scheme destabilizes a closed-loop system or degrades the control performance due to overflow.
This paper presents a novel controller representation based on input-output history data to implement an encrypted dynamic controller that operates without destabilization and performance degradation.
Implementation of this encrypted dynamic controller representation can be optimized via batching techniques to reduce the time and space complexities.
Furthermore, this study analyzes the stability and performance degradation of a closed-loop system caused by the effects of controller encryption.
A numerical simulation demonstrates the feasibility of the proposed encrypted control scheme, which inherits the control performance of the original controller at a sufficient level.
\end{abstract}


\section{Introduction}

The development of cloud computing technologies has promoted outsourcing computation of resource-limited devices, as well as data storage.
Control as a service (CaaS) is a concept of cloud-based control that outsources decision-making and monitoring of controlled devices to remote servers, and it is introduced to several control systems, such as automation~\cite{Hegazy15}, robotics~\cite{Vick15}, and automobiles~\cite{Esen15}.
CaaS has the advantages of scalability and efficiency in terms of energy and costs.
However, such control induces security concerns that confidential information of control recipes and privacy of controlled devices can be exposed to an untrusted third party.

Homomorphic encryption~\cite{Acar18} is the major approach for establishing secure outsourcing computation while maintaining the confidentiality of a computation process by allowing the direct evaluation of mathematical arithmetic on encrypted data without decryption.
An encrypted control framework~\cite{Darup21} was introduced in the literature on control engineering to apply homomorphic encryption for several control operations~\cite{Farokhi17,Kogiso15,Teranishi19_1,Kim19_2,Fritz19,Suh21,Darup18_2,Darup19_1,Alexandru21,Zhang21}.
Encrypted control can reduce the vulnerabilities induced by CaaS because the controller parameters and control signals over networks are encrypted while a controller server does not have a decryption key.

\subsection{Problem of encrypting dynamic controller}\label{sec:problem}

Most of the studies on encrypted control considered the encryption of static or linear controllers by using additive or multiplicative homomorphic encryption~\cite{Darup21,Farokhi17,Kogiso15}, where the encrypted controller states were decrypted on the plant side at every sampling period.
However, encrypting dynamic controllers without temporary decryption remains challenging because it causes an overflow of the states to be encrypted.
The majority of homomorphic encryption schemes operate using integers rather than real numbers, and thus controller states and other signals should be quantized before encryption.
To deal with this quantization, some studies used a binary representation of fractional numbers~\cite{Farokhi17,Darup19_1} or rounding of real numbers to the nearest element in a plaintext space after scaling~\cite{Kogiso15,Teranishi19_3}.
The former increases the number of bits for the representation by recursively updating the controller states with homomorphic operations~\cite{Murguia20_1}.
It is also inevitable for the latter that a value of the scaling parameter is accumulated for every homomorphic operation.
These effects induce overflow when a dynamic controller is naively encrypted because a plaintext space is a finite set.
Once overflow occurs, the decrypted state might be significantly different from the correct value, thereby easily inducing instability of the encrypted control systems.
Note that although some homomorphic encryption schemes, such as Cheon-Kim-Kim-Song (CKKS) encryption~\cite{Cheon17}, support floating-point number computation, such schemes include a quantization process with scaling factors in their encryption algorithms and rescaling to manage the factors.
The rescaling leads to overflow by recursive computation because it decreases the size of a ciphertext modulus along with resetting the accumulation of the factors.
Hence, quantization effects cannot be avoided in encrypted control systems.

To overcome this problem, the authors of~\cite{Murguia20_1} proposed a controller that resets its states at a constant period to clear an increase of the number of bits.
However, such reset operation obviously degrades the control performance.
Another approach in~\cite{Kim16} employed fully homomorphic encryption for encrypting dynamic controllers.
Fully homomorphic encryption enables to evaluate any arithmetic operations over a ciphertext space.
Thus, encrypted dynamic controllers can be implemented using fully homomorphic encryption because the accumulation of scaling parameters can be removed by division over a ciphertext space.
However, fully homomorphic encryption requires bootstrapping, which is a technique used to manage noise in a ciphertext.
Bootstrapping requires a large amount of computation time and resources to be performed, and therefore the realization of such encrypted control in real time is difficult in practice.
Note that, in fully homomorphic encryption, a small noise is injected into a ciphertext to guarantee security.
This noise grows every homomorphic operation, and the decryption result includes a large error when the noise reaches a certain size.
Hence, bootstrapping is essential for correct computation while ensuring security.

Recent studies~\cite{Kim22,Kim20_2} reformulated a dynamic controller by pole placement and similarity transformation so that its system matrix becomes an integer matrix.
This transformation makes quantization of the system matrix unnecessary, and thus the scaling parameter for the controller states does not accumulate.
Moreover, these studies regarded the effect of injected noises in Gentry-Sahai-Waters (GSW) encryption~\cite{Gentry13} as an external disturbance.
When the dynamic controller is stable under the disturbance, bootstrapping for the encryption scheme was found to be unnecessary for the infinite-time-horizon operation of the controller.
Although this approach is promising for implementing an encrypted dynamic controller without decrypting the controller states, it requires many computational resources to store and compute over a larger number of ciphertexts in the encrypted controller implementation.
To the best of our knowledge, the available fully homomorphic encryption in the approach is limited to the GSW encryption because the approach depends on the properties of the encryption scheme to manage the injected noises.
Furthermore, the reformulation of large-dimensional controllers sometimes fails due to the numerical instability of pole placement.

\subsection{Threat model and control system architecture}

This study considers the cloud-based control system shown in \figref{fig:setting}, where $\Enc$ and $\Dec$ are an encryptor and a decryptor, respectively.
The operator and sensor transmit a reference and plant output to the encrypted controller, respectively, while the signals are encrypted by the encryptors.
The encrypted controller computes a control input and returns it to the plant.
The decryptor decrypts the control input, and the actuator drives the plant using the decrypted control input.
Note that the plant has a public and secret key.
In contrast, the operator and encrypted controller have public and relinearization keys, respectively.
A relinearization key is a type of public key; its details are described later.
We assume an eavesdropper exists, and the cloud is honest-but-curious.
The eavesdropper aims to steal confidential information of the system by eavesdropping on signals over network links.
The honest-but-curious cloud tries to disclose some information from obtained data while following a correct protocol.

We employ a leveled fully homomorphic encryption scheme to protect the control system against the adversaries.
The encryption scheme can conceal both controller parameters and signals as opposed to additive and multiplicative homomorphic encryption schemes and does not require bootstrapping different from fully homomorphic encrytpion.
Additionally, some schemes based on ring learning with errors (RLWE)~\cite{Lyubashevsky10} provide a batching technique, which is beneficial for efficient computations.
Because of the advantages, some recent studies have used leveled fully homomorphic encryption to realize data-driven control~\cite{Alexandru20_2,Alexandru20_3} and reinforcement learning~\cite{Suh21}.
Furthermore, it is utilized for computing matrix multiplication to perform secure inference of machine learning~\cite{Gilad-Bachrach16,Hesamifard17,Juvekar18,Bourse18,Jiang18}.

\begin{figure}[!t]
    \centering
    \includegraphics[scale=1]{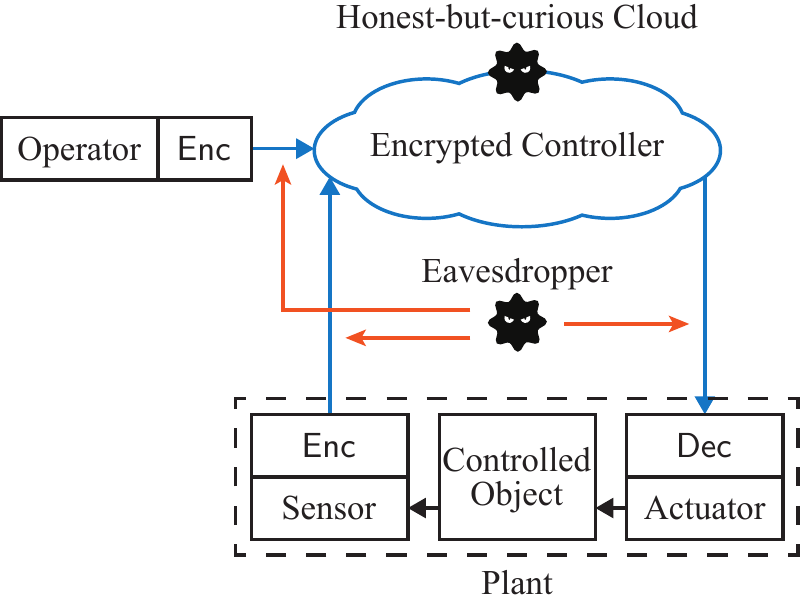}
    \caption{%
    Cloud-based encrypted control system under adversaries.
    The blue arrows are encrypted channels, and the red arrows illustrate eavesdropping attacks.
    }%
    \label{fig:setting}
    \vspace{-5mm}
\end{figure}

\subsection{Contribution}

This study contributes to establishing a general efficient framework for the secure outsourcing computation of dynamic controllers.
The intrinsic difficulty of encrypting a dynamic controller stems from the hardness of the recursive update of the encrypted controller states without decryption.
To solve this difficulty, we propose a novel controller representation using input/output history data of a controller to eliminate the controller states.
Furthermore, we clarify a condition so that an encrypted control system based on the proposed framework inherits the stability of the original control system under quantization effects due to encryption.
Additionally, the degree of performance degradation in the worst case is guaranteed as an upper bound of the output trajectory error between the encrypted and unencrypted control systems.
The feasibility of the proposed framework is demonstrated through a numerical simulation with a decentralized PI controller for a practical tank system~\cite{Johansson00}.

In contrast to the methods of~\cite{Kim22,Kim20_2}, our framework can be combined with any leveled fully homomorphic encryption because the proposed controller representation can solve the overflow and injected-noise management problems regardless of an employed encryption scheme.
According to this property, the batching technique of RLWE-based encryption can be easily used to reduce space and time complexities compared with the methods.
Additionally, transformation to the proposed controller representation is numerically stable even though the dimension of the original controller is large.

\subsection{Outline}

The remainder of this paper is organized as follows.
\secref{sec:preliminaries} introduces the notations and leveled fully homomorphic encryption used in this paper.
\secref{sec:representation} presents a novel controller representation using the input-output  history data for encrypted control.
\secref{sec:encryption} proposes an encrypted control scheme that operates for an infinite time horizon without overflow using the novel controller representation.
\secref{sec:analysis} analyzes the effects of quantization by encrypting the controller for stability and performance degradation of a closed-loop system.
\secref{sec:simulation} shows the results of a numerical simulation.
\secref{sec:conclusion} describes the conclusions of this study and future work.

\section{Preliminaries}\label{sec:preliminaries}

\subsection{Notation}\label{sec:notation}

The sets of real numbers, integers, and natural numbers are denoted by $\R$, $\Z$, and $\N$, respectively.
For $n>1$, $\Z_{n}$ denotes the set of integers $\{z\in\Z\mid -n/2< z\le n/2\}$.
A polynomial ring $R$ is defined as $R\coloneqq\Z[X]/(X^{N}+1)$, where $N$ is a power of $2$.
The set of polynomials in $R$ with coefficients in $\Z_{n}$ is denoted by $R_{n}$.
For $x\in\R$, the symbol $\round{\cdot}$ is defined by $\round{x}=\floor{x+1/2}$, where $\floor{\cdot}$ is the floor function.
The minimal residue of $z\in\Z$ modulo $n$ is denoted by $[z]_{n}$.
Similarly, $[\bfa]_{n}$ denotes the polynomial in $R_{n}$ given by applying $[\cdot]_{n}$ to all coefficients of $\bfa\in R$.
The sets of $n$-dimensional real-valued vectors and $m$-by-$n$ real-valued matrices are denoted by $\R^{n}$ and $\R^{m\times n}$, respectively.
The $n$-dimensional zero row vector and $m$-by-$n$ zero matrix are denoted by $\bfzero_{n}$ and $O_{m\times n}$, respectively.
The $i$th element of vector $v\in\R^{n}$ is denoted by $v_{i}$.
The $(i,j)$ entry of matrix $M\in\R^{m\times n}$ is denoted by $M_{ij}$.
The $\ell_{2}$ norm of $v$ and the induced $2$-norm of $M$ are denoted by $\|v\|$ and $\|M\|$, respectively.
The vectorization of $M$ is defined by $\vec(M)\coloneqq[M_{1}^{\top}\cdots M_{n}^{\top}]^{\top}$, where $M_{i}$ is the $i$th column vector of $M$, and $M_{i}^{\top}$ is the transpose matrix of $M_{i}$.
The spectral radius of $M$ is denoted by $\rho(M)$.

\subsection{Brakerski/Fan-Vercauteren encryption scheme}

This section provides an overview of the Brakerski/Fan-Vercauteren (BFV) leveled fully homomorphic encryption scheme~\cite{Fan12} used in this study.
The plaintext space of the BFV scheme is a polynomial ring, which is beneficial for the effective encryption/decryption of vector data.
The details of the security and algorithms are described in \appref{app:BFV}.

The BFV scheme consists of algorithms $\KeyGen$, $\Enc$, $\Dec$, $\Add$, and $\Mult$.
The key generation algorithm $\KeyGen(\lambda)$ takes a security parameter (i.e., key length) $\lambda\in\N$ and outputs a public key $\pk\in R_{Q}^{2}$, secret key $\sk\in R_{2}$, and relinearization key $\rlk\in R_{Q}^{2}$.
The public and secret keys are respectively used to encrypt a plaintext and decrypt a ciphertext.
The relinearization key is to be published and is required for the process of relinearization.
The encryption algorithm $\Enc(\pk,\bfm)$ takes a public key $\pk$ and a plaintext $\bfm\in R_{T}$ and outputs a ciphertext $\ct\in R_{Q}^{2}$.
Conversely, the decryption algorithm $\Dec(\sk,\ct)$ takes a secret key $\sk$ and a ciphertext $\ct$ and outputs a plaintext $\bfm$.
The ciphertexts of the BFV scheme must be correctly decrypted for all plaintexts in $R_{T}$ with valid parameters, namely $\Dec(\sk,\Enc(\pk,\bfm))=\bfm$, $\forall\bfm\in R_{T}$ for some $\lambda$ and for any $\pk$ and $\sk$ generated by $\KeyGen(\lambda)$.

The BFV scheme enables evaluation of both addition and multiplication over the ciphertext space.
The number of the evaluation depends on a security parameter $\lambda$.
The addition algorithm $\Add(\ct_{1},\ct_{2})$ takes ciphertexts $\ct_{1}=\Enc(\pk,\bfm_{1})$ and $\ct_{2}=\Enc(\pk,\bfm_{2})$ and satisfies $\Dec(\sk,\Add(\ct_{1},\ct_{2}))=[\bfm_{1}+\bfm_{2}]_{T}$.
Similarly, the multiplication algorithm $\Mult(\ct_{1},\ct_{2})$ computes the homomorphic multiplication.
Note that the number of elements in a ciphetext $\ct_{3}=\Mult(\ct_{1},\ct_{2})$ is three, namely $\ct_{3}\in R_{Q}^{3}$.
Hence, we cannot compute $\Add(\ct_{1},\ct_{3})$ and $\Add(\ct_{2},\ct_{3})$ because $\ct_{1}$ and $\ct_{2}$ are in $R_{Q}^{2}$.
The relinearization algorithm $\Relin(\rlk,\ct_{3})$ takes a relinearization key $\rlk$ and a ciphertext $\ct_{3}=\Mult(\ct_{1},\ct_{2})$ and outputs a ciphertext having two elements.
We assume that relinearization is performed after every homomorphic multiplication, and then $\Dec(\sk,\Mult(\ct_{1},\ct_{2}))=[\bfm_{1}\bfm_{2}]_{T}$ holds.
The security of the BFV scheme is based on the RLWE problem.
The problem is believed to be computationally hard to solve, even when a quantum computer is used~\cite{Lyubashevsky10}.
In the following, we omit the keys of the arguments for simplicity.

\subsection{Encoding and batching}

Although the controller parameters and signals are real-number matrices and vectors, respectively, the plaintext space of the BFV scheme is a polynomial ring with coefficients of integer modulo $T$.
Hence, the matrices and vectors should be converted into a polynomial before encryption.
To this end, we first consider the following encoder and decoder to convert a real number into an integer modulo $T$, and vice versa: $\Ecd_{\Delta}:\R\to\Z_{T}:x\mapsto[\round{x/\Delta}]_{T}$, $\Dcd_{\Delta}:\Z_{T}\to\R:z\mapsto\Delta z$, where $\Delta>0$ is a sensitivity for tuning the rounding errors caused by the encoding process, and the encoder and decoder operate for each element of vectors and matrices.
With proper sensitivity, the elements of the matrices and vectors can be encoded and decoded with the desired precision.
The quantization effects for stabilty and control performance are analyzed later.

Next, we consider transforming matrices and vectors of integer modulo $T$ to the corresponding polynomials.
One possible way for the transformation is to regard an element of the matrices and vectors as a polynomial of degree zero.
This trivial transformation does not require additional computation processes.
However, it is not efficient from the perspective of the total computation cost because an $m$-by-$n$ matrix or an $n$-dimensional vector have the same number of polynomials to represent the plaintext, namely $mn$ or $n$ polynomials.
The encryption and decryption algorithms should perform all polynomials, and thus a large amount of computation time and resources are required.

This study employs a batching technique for efficient transformation from matrices and vectors to polynomials.
Batching based on the Chinese remainder theorem (CRT) is a technique used in RLWE-based encryption to pack multiple integers into a single polynomial plaintext.
The CRT batching is effective for accelerating the computation of cryptosystems by allowing single instruction/multiple data (SIMD) operations for homomorphic evaluation~\cite{Smart14}.
The remainder of this section describes the CRT batching for the BFV scheme.

Suppose $T$ is a prime such that $T=1\bmod 2N$, where $N$ is a power of $2$ used for defining a polynomial ring $R$ in \secref{sec:notation}.
From the CRT, we have the ring isomorphism $R_{T}=\Z_{T}[X]/(X-\zeta)(X-\zeta^{3})\cdots(X-\zeta^{2N-1})\cong\Z_{T}[X]/(X-\zeta)\times\cdots\times\Z_{T}[X]/(X-\zeta^{2N-1})\cong\Z_{T}[\zeta]\times\cdots\times\Z_{T}[\zeta^{2N-1}]\cong\Z_{T}^{N}$~\cite{Chen17}, where $\zeta$ is the primitive $2N$th root of unity, that is, $\zeta^{2N}=1\bmod T$ and $\zeta^{i}\neq 1\bmod T$ for $0<i<2N$.
The CRT batching is constructed based on this isomorphism.
A canonical map from $R_{T}$ to $\Z_{T}^{N}$ is given as
\begin{equation}
    \bfm(X)\mapsto[\bfm(\zeta)\ \ \bfm(\zeta^{3})\ \ \cdots\ \ \bfm(\zeta^{2N-1})].
    \label{eq:dcd_CRT}
\end{equation}
The map \eqref{eq:dcd_CRT} can be represented by the nega-cyclic number theoretic transform (NTT) as follows~\cite{Longa16}:
\begin{align}
    &\sigma^{-1}:R_{T}\to\Z_{T}^{N}:\bfm=\sum_{i=0}^{N-1}m_{i}X^{i}\mapsto[z_{1}\ \cdots\ z_{N}], \label{eq:dcd_NTT} \\
    &\!\!\!\!\!\begin{bmatrix}
        z_{1} \\
        z_{2} \\
        \vdots \\
        \!z_{N}\!
    \end{bmatrix}\!\!=\!\!
    \left[\!
    \begin{bmatrix}
        1 & 1 & \cdots & 1 \\
        1 & \omega & \cdots & \omega^{N-1} \\
        \vdots & \vdots & \ddots & \vdots \\
        1 & \omega^{N-1} & \cdots & \omega^{(N-1)^{2}} \\
    \end{bmatrix}\!\!\!
    \left(\!
    \begin{bmatrix}
        1 \\
        \zeta \\
        \vdots \\
        \!\zeta^{N-1}\!
    \end{bmatrix}\!\!\odot\!\!
    \begin{bmatrix}
        m_{0} \\
        m_{1} \\
        \vdots \\
        \!m_{N-1}\!
    \end{bmatrix}
    \!\right)
    \!\!\right]_{\!T}\!\!, \nonumber
\end{align}
where $\odot$ denotes the Hadamard product, and $\omega$ is the primitive $N$th root of unity.
Additionally, the inverse transformation of \eqref{eq:dcd_NTT} with the inverse NTT is given as
\begin{align}
    &\sigma:\Z_{T}^{N}\to R_{T}:[z_{1}\ \cdots\ z_{N}]\mapsto\bfm=\sum_{i=0}^{N-1}m_{i}X^{i}, \label{eq:ecd_NTT} \\
    &\!\!\!\!\begin{bmatrix}
        m_{0} \\
        m_{1} \\
        \vdots \\
        \!m_{N-1}\!
    \end{bmatrix}\!\!=\!\!
    \left[\!
    \begin{bmatrix}
        1 \\
        \xi \\
        \vdots \\
        \!\xi^{N-1}\!
    \end{bmatrix}\!\!\odot\!\!
    \left(\!\!
    \cfrac{1}{N}\!
    \begin{bmatrix}
        1 & 1 & \cdots & 1 \\
        1 & \pi & \cdots & \pi^{N-1} \\
        \vdots & \vdots & \ddots & \vdots \\
        1 & \pi^{N-1} & \cdots & \pi^{(N-1)^{2}} \\
    \end{bmatrix}\!\!\!\!
    \begin{bmatrix}
        z_{1} \\
        z_{2} \\
        \vdots \\
        \!z_{N}\!
    \end{bmatrix}
    \!\right)
    \!\!\right]_{\!T}\!\!, \nonumber
\end{align}
where $\xi=[\zeta^{-1}]_{T}$ and $\pi=[\omega^{-1}]_{T}$.
Consequently, we obtain the map \eqref{eq:ecd_NTT} for packing multiple integers into a single polynomial and the map \eqref{eq:dcd_NTT} for unpacking the polynomial.

\section{Input-Output History Feedback Controller}\label{sec:representation}

This section proposes a novel controller representation to implement an encrypted dynamic controller.
The proposed scheme represents any linear time-invariant controller without controller states, instead of using the history of inputs and outputs of the controller.

Given a discrete-time system
\begin{equation}
    \left\{
    \begin{aligned}
        x_{t+1}&=A_{p}x_{t}+B_{p}u_{t}+w_{t}, \\
        y_{t}&=C_{p}x_{t}+v_{t},
    \end{aligned}
    \right.
    \label{eq:plant}
\end{equation}
where $t\in\N$ is a time, $x\in\R^{n}$ is a state, $u\in\R^{m}$ is an input, $y\in\R^{\ell}$ is an output, $w\in\R^{n}$ is a process noise, and $v\in\R^{\ell}$ is a measurement noise.
$(A_{p},B_{p})$ and $(A_{p},C_{p})$ are controllable and observable, respectively.
This study considers encrypting the following linear time-invariant controller based on history data to control the plant \eqref{eq:plant}:
\begin{equation}
    \left\{
    \begin{aligned}
        z_{t+1}&=Az_{t}+By_{t}+Er_{t}, \\
        u_{t}&=Cz_{t}+Dy_{t}+Fr_{t},
    \end{aligned}
    \right.
    \label{eq:controller}
\end{equation}
where $z\in\R^{p}$ is a controller state, and $r\in\R^{q}$ is a reference input.
We assume that the pair $(A,C)$ is observable without loss of generality.
If the pair is not observable, the controller can be reconstructed as minimal realization.
This section presents another representation of \eqref{eq:controller} without using its state $z$, called the \textit{input-output history feedback controller (IOHFC) representation}, in order to enable an encrypted controller of \eqref{eq:controller} to operate for an infinite time horizon without overflow.

Let $[d_{k}]_{t_{2}}^{t_{1}}\coloneqq[d_{t_{1}}^{\top}\ \cdots\ d_{t_{2}}^{\top}]^{\top}$ be a stacked vector of time-series data $d_{k}$ for $t_{1}\le k\le t_{2}$.
With this notation, the following equations are obtained from \eqref{eq:controller}:
\begin{align}
    z_{t}&=A^{L}z_{t-L}+R_{L}[y_{k}]_{t-1}^{t-L}+S_{L}[r_{k}]_{t-1}^{t-L}, \label{eq:stacked_state} \\
    [u_{k}]_{t-1}^{t-L}&=V_{L}z_{t-L}+H_{L}[y_{k}]_{t-1}^{t-L}+J_{L}[r_{k}]_{t-1}^{t-L}, \label{eq:stacked_input}
\end{align}
where $L>0$ is a data length, $V_{L}\coloneqq[C^{\top}\ \cdots\ (CA^{L-1})^{\!\top}]^{\!\top}$ $\in\R^{Lm\times p}$, $R_{L}\coloneqq[A^{L-1}B\ \cdots\ B]\in\R^{p\times L\ell}$, $S_{L}\coloneqq[A^{L-1}E\ \cdots\ E]\in\R^{p\times Lq}$, and
\begin{align*}
    H_{L}&\coloneqq
    \begin{bmatrix}
        D         & O         & \cdots & O      & O \\
        CB        & D         & \cdots & O      & O \\
        \vdots    & \vdots    & \ddots & \vdots & \vdots \\
        CA^{L-2}B & CA^{L-1}B & \cdots & CB     & D
    \end{bmatrix}\in\R^{Lm\times L\ell}, \\
    J_{L}&\coloneqq
    \begin{bmatrix}
        F         & O         & \cdots & O      & O \\
        CE        & F         & \cdots & O      & O \\
        \vdots    & \vdots    & \ddots & \vdots & \vdots \\
        CA^{L-2}E & CA^{L-1}E & \cdots & CE     & F
    \end{bmatrix}\in\R^{Lm\times Lq}.
\end{align*}
Assume that $L\ge p$ is chosen to satisfy $\rank V_{L}=p$.
Then, there exists the Moore-Penrose inverse $V_{L}^{+}$ of $V_{L}$ such that $V_{L}^{+}V_{L}=I$ because $V_{L}$ is full column rank.
Thus, it follows from \eqref{eq:stacked_input} that
\begin{equation}
    z_{t-L}=V_{L}^{+}[u_{k}]_{t-1}^{t-L}-V_{L}^{+}H_{L}[y_{k}]_{t-1}^{t-L}-V_{L}^{+}J_{L}[r_{k}]_{t-1}^{t-L}. \label{eq:x_tL}
\end{equation}
By substituting \eqref{eq:x_tL} into \eqref{eq:stacked_state}, the controller state $z$ at time $t$ can be represented as
\begin{align*}
    z_{t}&\!=\!(S_{L}-A^{L}V_{L}^{+}J_{L})[r_{k}]_{t-1}^{t-L}+(R_{L}-A^{L}V_{L}^{+}H_{L})[y_{k}]_{t-1}^{t-L} \\
    &\quad+A^{L}V_{L}^{+}[u_{k}]_{t-1}^{t-L}.
\end{align*}
Hence, the input $u$ can be computed as follows.
\begin{equation}
    u_{t}=Kd_{t},
    \label{eq:controller_data}
\end{equation}
where
\begin{align*}
    K&\coloneqq
    \Bigl[C(S_{L}-A^{L}V_{L}^{+}J_{L})\ \ F\ \ C(R_{L}-A^{L}V_{L}^{+}H_{L}) \\
    &\quad\quad D\ \ CA^{L}V_{L}^{+}\Bigr], \\
    d _{t}&\coloneqq
    \begin{bmatrix}
        ([r_{k}]_{t}^{t-L})^{\top} & ([y_{k}]_{t}^{t-L})^{\top} & ([u_{k}]_{t-1}^{t-L})^{\top}
    \end{bmatrix}^{\top}.
\end{align*}
The controller \eqref{eq:controller_data} is another representation of the controller \eqref{eq:controller} based on the history data and current output/reference, and their control inputs $u_{t}$ are the same after $L$ samples.
Consequently, we obtain the following theorem.

\begin{theorem}\label{theorem:iohfc}
    For the linear time-invariant controller \eqref{eq:controller}, there exists an IOHFC \eqref{eq:controller_data} such that its control input exactly matches that of \eqref{eq:controller} for all $t\ge L$.
    Furthermore, if $z_{t}=0$ for all $t\le 0$, the control inputs of \eqref{eq:controller} and \eqref{eq:controller_data} are identical for all $t\ge 0$.
\end{theorem}

A schematic picture of the IOHFC is illustrated in \figref{fig:controller_data}.
The controller has three queues of lengths $L+1$ and $L$ to store history data.
A reference input $r$ and output $y$ are respectively transmitted to the controller from an operator and a plant at time $t$ and added to the back end of the queues.
Then, the control input $u$ of the controller \eqref{eq:controller_data} is simply computed by the product between a controller gain $K$ and a data vector $d$ constructed from the history data.
The control input is returned to the plant and appended to the back end of the queue simultaneously.

\begin{remark}
    The proposed controller representation is a form of a vector autoregressive model with exogenous variables.
    Some studies have used the model to represent dynamical systems~\cite{Jansson03,Chiuso07}.
    However, to the best of our knowledge, few studies have applied the model to dynamic feedback controllers.
    \thmref{theorem:iohfc} reveals that a dynamic controller can be represented by history data instead of the controller state without performance degradation.
\end{remark}

\begin{remark}
    The IOHFC representation can be applied to not only controllers but also any linear time-invariant systems.
    For example, the Kalman filter for \eqref{eq:plant}, $\hat{x}_{t+1}=(A_{p}-GC_{p})\hat{x}_{t}+B_{p}u_{t}+Gy_{t}$, $\hat{y}_{t}=C_{p}\hat{x}_{t}$, can be representated as $\hat{y}_{t}=\Sigma[([y_{k}]_{t}^{t-L})^{\top}\ ([u_{k}]_{t-1}^{t-L})^{\top}]^{\top}$, where $\hat{x}$ is an estimated state, $\hat{y}$ is an estimated output, $G$ is a Kalman gain, and $\Sigma$ is an appropriate matrix obtained by the IOHFC representation.
    Hence, secure outsourcing computation of forecasting, filtering, and sensor fusion associated with dynamics can also be realized by using the proposed representation.
\end{remark}

\begin{figure}[!t]
    \centering
    \includegraphics[scale=1]{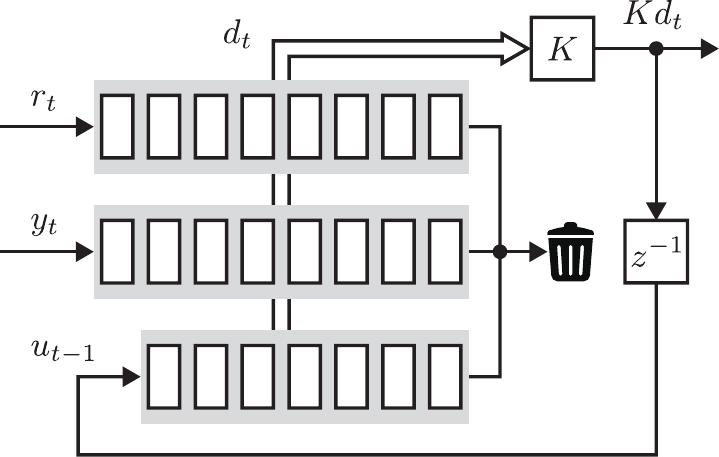}
    \caption{Schematic picture of the IOHFC ($L=7$).}
    \label{fig:controller_data}
    \vspace{-5mm}
\end{figure}

\section{Encrypting IOHFC}\label{sec:encryption}

This section provides an efficient matrix-vector multiplication algorithm over encrypted data using the encoding and batching technique in \secref{sec:preliminaries}.
Furthermore, an algorithm for the implementation of an encrypted IOHFC is proposed based on the matrix-vector multiplication algorithm to allow the controller to operate for an infinite time horizon without overflow.

\subsection{Matrix-vector multiplication by SIMD operations}

We present a method for secure matrix-vector multiplication by SIMD operations~\cite{Gilad-Bachrach16,Kim18}.
With the CRT batching, element-wise addition and multiplication between two vectors $v_{1},v_{2}\in\Z_{T}^{N}$ can be evaluated by the SIMD operations over the ciphertext space as follows:
\begin{align*}
    \sigma^{-1}\circ\Dec\circ\Add(\Enc\circ\sigma(v_{1}),\Enc\circ\sigma(v_{2}))&=[v_{1}+v_{2}]_{T}, \\
    \sigma^{-1}\circ\Dec\circ\Mult(\Enc\circ\sigma(v_{1}),\Enc\circ\sigma(v_{2}))&=[v_{1}\odot v_{2}]_{T}.
\end{align*}
Moreover, the BFV scheme with the CRT batching allows permutations of plaintext slots, that is, the positions of the elements of an integer vector, by using the Galois automorphism sending a polynomial $\bfm(X)\in R_{T}$ to $\bfm(X^{2i-1})$~\cite{Chen17}.
We define this permutation of shifting one slot to the left as $\Rotate:R_{Q}^{2}\to R_{Q}^{2}:\ct\mapsto\ct'$, and it satisfies $\sigma^{-1}\circ\Dec\circ\Rotate\circ\Enc\circ\sigma([z_{1}\ z_{2}\ \cdots\ z_{N}])=[z_{2}\ \cdots\ z_{N}\ z_{1}]$, where $z_{i}\in\Z_{T}$ for $1\le i\le N$.

\algref{alg:mat_mult} is the method combined with $\Add$, $\Mult$, and $\Rotate$ to compute multiplication between a matrix and vector, and \figref{fig:mat_mult} is the illustration of the algorithm.
The matrix $M$ is embedded in the first $mn$ elements of the temporary row vector $z_{1}$ such that each row of the matrix lines up (line $2$).
Similarly, the vector $v$ is copied $m$ times, and then the vectors are also embedded in $z_{2}$.
These processes are shown in \figref{fig:mat_mult}(b).
The temporary vectors are packed into single polynomials and encrypted (line $4$).
Using the SIMD operations, each element of the vectors is multiplied over the ciphertext space, and then the resultant vector is added with the rotation of itself $n-1$ times (lines $6$--$9$).
The computed ciphertext is decrypted and unpacked to $z_{3}$, and the $(i+1)$th elements of $z_{3}$ for $0\le i\le (m-1)n$ are extracted to construct the target vector $[Mv]_{T}$ (lines $11$--$12$).
The SIMD operations of the vectors and the construction of the target vector are shown in \figref{fig:mat_mult}(c).

We employ $\Enc_{\Delta}\coloneqq\Enc\circ\sigma\circ\Ecd_{\Delta}$ and $\Dec_{\Delta}\coloneqq\Dcd_{\Delta}\circ\sigma^{-1}\circ\Dec$ for the encryption and decryption of a real-valued vector in the following.

\begin{figure}[!t]
    \begin{algorithm}[H]
        \caption{Secure matrix-vector multiplication by SIMD operations.}
        \label{alg:mat_mult}
        \begin{algorithmic}[1]
            \Require $M\in\Z_{T}^{m\times n}$, $v\in\Z_{T}^{n}$
            \Ensure $[Mv]_{T}$
            \State Let $z_{1}$, $z_{2}$, and $z_{3}$ be row vectors in $\Z_{T}^{N}$
            \State $z_{1}\gets[\vec(M^{\top})^{\top}\ \bfzero_{N-mn}]$, $z_{2}\gets[v^{\top}\ \cdots\ v^{\top}\ \bfzero_{N-mn}]$
            \State \# Encryption
            \State $\ct_{1}\gets\Enc\circ\sigma(z_{1})$, $\ct_{2}\gets\Enc\circ\sigma(z_{2})$
            \State \# Multiplication over the ciphertext space
            \State $\ct_{3}\gets\Mult(\ct_{1},\ct_{2})$
            \While{$n-1$ times}
                \State $\ct_{3}\gets\Add(\ct_{3},\Rotate(\ct_{3}))$
            \EndWhile
            \State \# Decryption
            \State $z_{3}\gets\sigma^{-1}\circ\Dec(\ct_{3})$
            \State \Return $[z_{3,1}\ z_{3,n+1}\ \cdots\ z_{3,(m-1)n+1}]^{\top}$
        \end{algorithmic}
    \end{algorithm}
    \vspace{-10mm}
\end{figure}

\begin{figure*}[!t]
    \centering
    \includegraphics[scale=1]{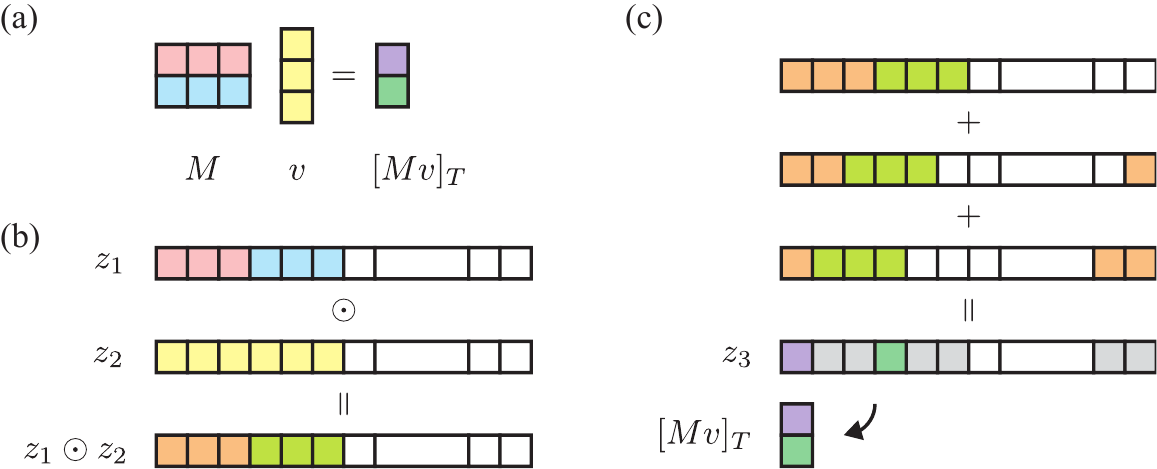}
    \caption{Illustration of secure matrix-vector multiplication over the ciphertext space. (a) Matrix $M\in\Z_{T}^{2\times3}$, vector $v\in\Z_{T}^{3}$, and target vector $[Mv]_{T}\in\Z_{T}^{2}$. (b) The elements of $M$ and $v$ are embedded in corresponding $N$-dimensional vectors $z_{1}$ and $z_{2}$, respectively. The white boxes of the vectors contain zero. Element-wise multiplication $z_{1}\odot z_{2}$ between the vectors is computed. (c) The computed vector is added with the rotation of itself three times to obtain $z_{3}$. The gray boxes of $z_{3}$ are wasted data. The target vector is constructed from the first and forth elements of $z_{3}$.}
    \label{fig:mat_mult}
    \vspace{-3mm}
\end{figure*}

\subsection{Encrypted IOHFC with input re-encryption}

Encrypting the IOHFC \eqref{eq:controller_data} may appear to be straightforward because leveled fully homomorphic encryption enables the evaluation of both multiplication and addition over a ciphertext space.
However, the IOHFC cannot be directly implemented in an encrypted fashion, even though \algref{alg:mat_mult} is used.
This is because the history data $d$ cannot be updated recursively when the vector is encrypted to a single ciphertext.
One may think that history data can be updated by multiplying a masking vector to the ciphertext of $d$ and rotating the masked ciphertext.
Unfortunately, this approach causes overflow due to the increase of ciphertext noise because the masking vector must be multiplied every sampling time.
Moreover, the ciphertext cannot be altered into another one corresponding to an updated plaintext vector without decryption.

To overcome this problem, this study splits the controller gain of the IOHFC \eqref{eq:controller_data} into block matrices for each time step as follows:
\begin{align}
    u_{t}
    &=\sum_{i=0}^{L}K_{r,i}r_{t-i}+\sum_{i=0}^{L}K_{y,i}y_{t-i}+\sum_{i=0}^{L-1}K_{u,i}u_{t-(i+1)}, \nonumber \\
    &=\sum_{i=0}^{L}\hat{K}_{i}\hat{d}_{i}, \label{eq:controller_data_block}
\end{align}
where
\begin{align*}
    \hat{K}_{i}&=
    \left\{
    \begin{alignedat}{2}
        &\begin{bmatrix}
            K_{r,i} & K_{y,i} & K_{u,i}
        \end{bmatrix}, &\quad &i=0,\dots,L-1, \\
        &\begin{bmatrix}
            K_{r,i} & K_{y,i} & O_{m\times m}
        \end{bmatrix}, &\quad &i=L,
    \end{alignedat}
    \right. \\
    \hat{d}_{i}&=
    \left\{
    \begin{alignedat}{2}
        &\begin{bmatrix}
            r_{t-i}^{\top} & y_{t-i}^{\top} & u_{t-i}^{\top}
        \end{bmatrix}^{\top}, &\quad &i=0,\dots,L-1, \\
        &\begin{bmatrix}
            r_{t-i}^{\top} & y_{t-i}^{\top} & \bfzero_{m}
        \end{bmatrix}^{\top}, &\quad &i=L,
    \end{alignedat}
    \right. \\
    K_{r,i}&=K_{1:m,qi+1:q(i+1)}, \\
    K_{y,i}&=K_{1:m,q(L+1)+\ell i+1:q(L+1)+\ell(i+1)}, \\
    K_{u,i}&=K_{1:m,(q+\ell)(L+1)+mi+1:(q+\ell)(L+1)+m(i+1)},
\end{align*}
and $K_{a:b,c:d}$ is a block matrix obtained by slicing the $a$th to $b$th rows and the $c$th to $d$th columns of $K$.
\figref{fig:modified_controller_data} depicts a schematic picture of the modified IOHFC \eqref{eq:controller_data_block} that has a queue of length $L$.
The current reference $r_{t}$ and output $y_{t}$ are appended to the back end of the queue.
Then, the data $\hat{d}_{i}$ in each slot of the queue is multiplied by the gain $\hat{K}_{i}$ and aggregated to obtain a control input $u_{t}$.
The obtained control input is used to update the data $\hat{d}_{L-1}$ stored in the second slot from the back of the queue.
Note that the products between $\hat{K}_{i}$ and $\hat{d}_{i}$ and their aggregation can be computed over a ciphertext space by applying \algref{alg:mat_mult}.

\begin{figure}[!t]
    \centering
    \includegraphics[scale=1]{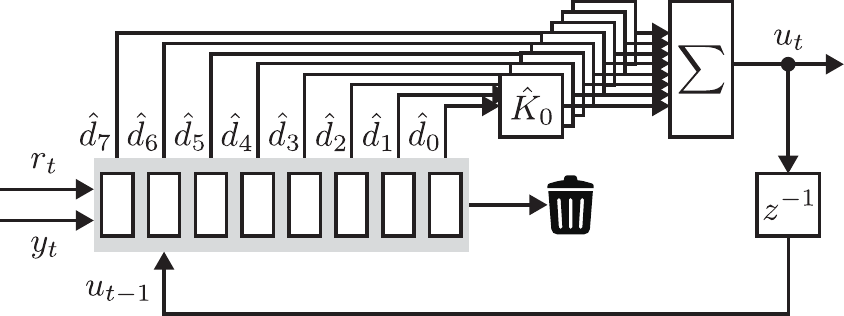}
    \caption{Schematic picture of the modified IOHFC ($L=7$).}
    \label{fig:modified_controller_data}
\end{figure}

Here, we need to address another problem due to accumulation of sensitivities and an increase of noise in a ciphertext stored in the controller.
A control input $u_{t}$ in \figref{fig:modified_controller_data} is added to the queue and recursively used to the next $L-1$ times the control input computations.
This operation accumulates the sensitivity $\Delta_{K}$, which is used for encoding the controller gain, and increases the noise of the ciphertext in the second slot from the back end of the queue.
In such a case, the encrypted controller cannot operate for an infinite time horizon due to overflow, as discussed in \secref{sec:problem}.

This study considers \textit{input re-encryption}, as shown in \figref{fig:ecs}, to solve the problem of accumulation of the sensitivity.
In the figure, the sensitivities in $\Enc,\Dec$ are omitted for simplicity.
The overall processes of the encrypted control system are described in \algref{alg:enc_control_system}.
Before operating the encrypted controller, a designer creates and distributes keys to the plant, operator, and controller (line $2$).
Additionally, he/she initializes the ciphertexts of the controller gain matrices $\ct_{K}$ and queue $\ct_{d}$, which stores ciphertexts of history data (lines $3$--$6$).
The operator and plant respectively pack and encrypt the current reference input $r_{t}$ and output $y_{t}$ into single ciphertexts and transmit them to the controller, and then controller adds the ciphertexts to the back end of the queue. (lines $8$--$13$).
The controller evaluates \eqref{eq:controller_data_block} over the ciphertext space with encrypted controller gains and encrypted history data using the same methodology as \algref{alg:mat_mult} (lines $14$--$21$).
The controller updates the queue and returns the computed ciphertext to the plant, and then the plant recovers input $u_{t}$ by decrypting the ciphertext (lines $22$--$28$).
Subsequently, the input is re-encrypted and transmitted to the controller, and then the controller adds the ciphertext to the second slot from the back end of the queue (lines $29$--$31$).

\begin{figure}[!t]
    \centering
    \includegraphics[scale=1]{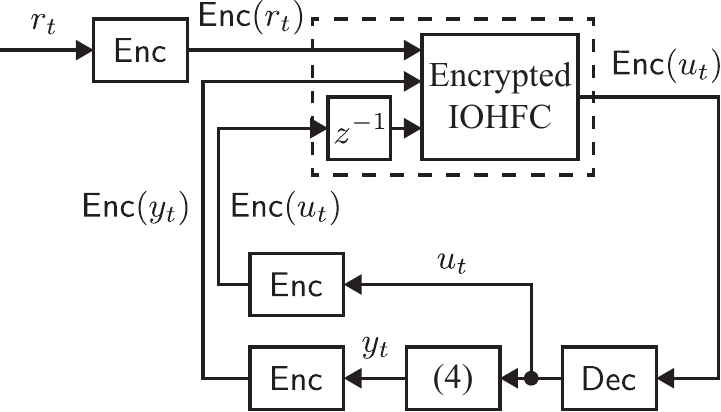}
    \caption{Encrypted control system with IOHFC.}
    \label{fig:ecs}
    \vspace{-5mm}
\end{figure}

\begin{figure}[!t]
    \vspace{-3mm}
    \begin{algorithm}[H]
        \caption{Implementation of the encrypted IOHFC with input re-encryption.}
        \label{alg:enc_control_system}
        \begin{algorithmic}[1]
            \Require $\lambda$, $K$, $r_{t}$, $y_{t}$, $\Delta_{K}$, $\Delta_{d}$, $L$, $q$, $\ell$, $m$, $h=q+\ell+m$
            \Ensure $u_{t}$
            \State \# Preprocessing of designer
            \State $(\pk,\sk,\rlk)\gets\KeyGen(\lambda)$, and transmit $(\pk,\sk)$, $\pk$, and $\rlk$ to the plant, operator, and controller, respectively
            \For{$i=0,\dots,L$}
                \State $\ct_{K}[i]\gets\Enc_{\Delta_{K}}([\vec(\hat{K}_{i}^{\top})^{\top}\ \bfzero_{N-mh}])$
                \State $\ct_{d}[i]\gets\Enc_{\Delta_{d}}(\bfzero_{N})$
            \EndFor
            \Loop
                \State \# Operator transmits reference input to controller
                \State $\ct_{r}\gets\Enc_{\Delta_{d}}([r_{t}^{\top}\ \bfzero_{\ell}\ \bfzero_{m}\ \cdots\ r_{t}^{\top}\ \bfzero_{\ell}\ \bfzero_{m}\ \bfzero_{N-mh}])$
                \State $\ct_{d}[L]\gets\Add(\ct_{d}[L],\ct_{r})$
                \State \# Plant transmits output to controller
                \State $\ct_{y}\gets\Enc_{\Delta_{d}}([\bfzero_{q}\ y_{t}^{\top}\ \bfzero_{m}\ \cdots\ \bfzero_{q}\ y_{t}^{\top}\ \bfzero_{m}\ \bfzero_{N-mh}])$
                \State $\ct_{d}[L]\gets\Add(\ct_{d}[L],\ct_{y})$
                \State \# Controller returns input ciphertext $\ct$ to plant
                \State $\ct\gets\Mult(\ct_{K}[0],\ct_{d}[0])$
                \For{$i=1,\dots,L$}
                    \State $\ct\gets\Add(\ct,\Mult(\ct_{K}[i],\ct_{d}[i]))$
                \EndFor
                \While{$h-1$ times}
                    \State $\ct\gets\Add(\ct,\Rotate(\ct))$
                \EndWhile
                \State \# Controller updates history data
                \For{$i=0,\dots,L-1$}
                    \State $\ct_{d}[i]\gets\ct_{d}[i+1]$
                \EndFor
                \State \# Plant recovers input
                \State $w\gets\Dec_{\Delta_{K}\Delta_{d}}(\ct)$
                \State $u_{t}\gets[w_{1}\ w_{h+1}\ \cdots\ w_{(m-1)h+1}]^{\top}$
                \State \# Plant transmits re-encrypted input to controller
                \State $\ct_{u}\gets\Enc_{\Delta_{d}}([\bfzero_{q}\ \bfzero_{\ell}\ u_{t}^{\top}\ \cdots\ \bfzero_{q}\ \bfzero_{\ell}\ u_{t}^{\top}\ \bfzero_{N-mh}])$
                \State $\ct_{d}[L-1]\gets\Add(\ct_{d}[L-1],\ct_{u})$
            \EndLoop
            \end{algorithmic}
    \end{algorithm}
    \vspace{-10mm}
\end{figure}

Owing to input re-encryption, the sensitivity of each element in the queue $\ct_{d}$ is $\Delta_{d}$, which is used for encoding the history data, even though the sensitivity of $\ct$ is $\Delta_{K}\Delta_{d}$.
Moreover, ciphertexts appended to the back end of the queue are always fresh ciphertexts.
Thus, the sensitivity and noise in each ciphertext stored in the queue do not accumulate and increase, and the encrypted control system with the IOHFC can operate for an infinite time horizon without overflow.

\begin{remark}
    The number of executed algorithms of the BFV scheme within a sampling period is listed in \tabref{tab:num_alg}, where $h=q+\ell+m$.
    Despite the fact that the reference input $r$, output $y$, and input $u$ are vectors, the operator and the plant respectively execute encryption only once and twice before sending them to the controller, and the plant also performs decryption once after receiving the ciphertext from the controller.
    Furthermore, the number of homomorphic multiplications computed by the controller does not depend on the dimensions of the vectors.
    This is because the vectors are packed into single ciphertexts by the CRT batching.
\end{remark}

\begin{table}[!t]
    \centering
    \caption{Number of Algorithms Executed Within a Sampling Period}
    \begin{tabular}{l|ccccc}
        \hline
                       & $\Enc$ & $\Dec$ & $\Add$  & $\Mult$ & $\Rotate$ \\ \hline
        Operator   & $1$    & --     & --      & --      & --        \\
        Plant      & $2$    & $1$    & --      & --      & --        \\
        Controller & --     & --     & $L+h+2$ & $L+1$   & $h-1$     \\ \hline 
    \end{tabular}
    \label{tab:num_alg}
    \vspace{-5mm}
\end{table}

\begin{remark}
    In~\cite{Kim22,Kim20_2}, the dynamic controller \eqref{eq:controller} was transformed to the form
    \[
        \left\{
        \begin{aligned}
            z'_{t+1} &= M (A - GC) M^{-1} z'_t + M (B - GD) y_t \\ 
                     &\quad + M (E - GF) r_t + MG u_t, \\
            u_t      &= CM^{-1} z'_{t} + D y_t + F r_t, \quad z'_0 = M z_0
        \end{aligned}
        \right.
    \]
    by appropriately choosing $G\in\R^{p\times m}$ and $M\in\R^{p\times p}$ so that $M(A-GC)M^{-1}\in\Z^{p\times p}$ and encrypted by using the GSW encryption.
    Thus, $[(p+q+\ell+m)(N+1)+\{p(p+q+\ell+m)+m(p+q+\ell)\}b(N+1)^{2}]\log_{2}Q$~bits memory is required to naively implement the encrypted controller because each element of the signals and controller parameters are respectively encrypted to an element in $\Z_{Q}^{N+1}$ and $\Z_{Q}^{(N+1)\times b(N+1)}$, where $b\in\N$.
    In contrast, the encrypted IOHFC requires $4(p+1)N\log_{2}Q$~bits memory because $\hat{K}_{i}$ and $\hat{d}_{i}$ for $i=0,\dots,L$ are encrypted to elements in $R_{Q}^{2}$, and $L=p$ is the smallest choice of data length.
    Note that, in a naive implementation, an element in $R_{Q}^{2}$ can be represented by two $N$-dimensional vectors of which coefficients are in $\Z_{Q}$.
    As an example, if the parameters are set to $p=q=\ell=m=2$, $N=4096$, $\log_2 Q=109$, and $b=3$, then the memory sizes of the conventional and our methods are almost $17.89$~GB and $654$~KB, respectively.
    Moreover, although the conventional method takes $p(p+q+\ell+m)+m(p+q+\ell)$ multiplications and $p(p+q+\ell+m)+m(p+q+\ell)-(4p+3m)$ additions for each time step, our method takes $p+1$ multiplications, $p+q+\ell+m+2$ additions, and $q+\ell+m-1$ rotations.
    The computation time of multiplication is much longer than addition and rotation.
    Therefore, our method can improve the time and space complexities of the conventional method.
\end{remark}

\section{Analysis of Quantization Effects}\label{sec:analysis}
\color{black}

This section analyzes the stability and performance degradation caused by quantization in encrypted control systems when $w_{t}=v_{t}=0$ for simplicity.
Let $\Qtz_{\Delta}$ be the composite mapping of $\Enc_{\Delta}$ and $\Dec_{\Delta}$, then $\Qtz_{\Delta}=\Dec_{\Delta}\circ\Enc_{\Delta}=\Dcd_{\Delta}\circ\sigma^{-1}\circ\Dec\circ\Enc\circ\sigma\circ\Ecd_{\Delta}=\Dcd_{\Delta}\circ\Ecd_{\Delta}$.
Thus, the map $\Qtz_{\Delta}$ behaves as a quantizer.
With the quantizer, the decrypted input of the encrypted IOHFC is equivalent to
\begin{equation}
    u_{t}\!=\!\sum_{i=0}^{L}\!\Qtz_{\Delta_{K}}\!(\hat{K}_{i})\Qtz_{\Delta_{d}}\!(\hat{d}_{i})\!=\!\Qtz_{\Delta_{K}}\!(K)\Qtz_{\Delta_{d}}\!(d_{t})\!=\!\bar{K}\bar{d}_{t},
    \label{eq:qtz_controller}
\end{equation}
where $\bar{K}\coloneqq\Qtz_{\Delta_{K}}(K)$ and $\bar{d}\coloneqq\Qtz_{\Delta_{d}}(d)$.
Quantization errors caused by the quantization are respectively bounded from above by
\begin{align}
    \|\tldd\|&\le\sqrt{(L+1)(q+\ell)+Lm}\Delta_{d}/2\eqqcolon\eta_{d}, \label{eq:bound_v} \\
    \|\tldK\|&\le\sqrt{(L+1)(q+\ell)m+Lm^{2}}\Delta_{K}/2\eqqcolon\eta_{K}, \label{eq:bound_K}
\end{align}
where $\tldd\coloneqq\bar{d}-d$ and $\tldK\coloneqq\bar{K}-K$.
The quantized controller \eqref{eq:qtz_controller} induces destabilization and performance degradation of the control system when the sensitivities are not sufficiently small.
We show a condition for maintaining stability even after quantization and estimating the degree of performance degradation in the following.
To this end, rewrite the system \eqref{eq:plant} as
\begin{equation}
    \left\{
    \begin{aligned}
        \sfx_{t+1}&=\sfA\sfx_{t}+\sfB u_{t}+\sfE r_{t}, \\
        d_{t}&=\sfC_{1}\sfx_{t}+\sfF r_{t},\quad y_{t}=\sfC_{2}\sfx_{t},
    \end{aligned}
    \right.
    \label{eq:plant_data}
\end{equation}
where $\sfx_{t}=[([r_{k}]_{t-1}^{t-L})^{\top}\ \bfzero_{q}\ ([x_{k}]_{t}^{t-L})^{\top}\ ([u_{k}]_{t-1}^{t-L})^{\top}]^{\top}$,
\begin{align*}
    \sfA&\!=\!\diag{\sfA_{1},\sfA_{2},\sfA_{3}}, \\
    \sfA_{1}&\!=\!
    \begin{bmatrix}
        O_{(L-1)q\times q} & I_{(L-1)q}         & O_{(L-1)q\times q} \\
        O_{2q\times q}     & O_{2q\times(L-1)q} & O_{2q\times q}
    \end{bmatrix}\!, \\
    \sfA_{2}&\!=\!
    \begin{bmatrix}
        O_{Ln\times n} & I_{Ln} \\
        O_{n\times Ln} & A_{p}
    \end{bmatrix}\!,\ 
    \sfA_{3}\!=\!
    \begin{bmatrix}
        O_{(L-1)m\times m} & I_{(L-1)m} \\
        O_{m\times m}      & O_{m\times(L-1)m}
    \end{bmatrix}\!, \\
    \sfB&\!=\!
    \begin{bmatrix}
        O_{((L+1)q+Ln)\times m}^{\top} & B_{p}^{\top} & O_{(L-1)m\times m}^{\top} & I_{m}
    \end{bmatrix}^{\top}\!, \\
    \sfC_{1}&\!=\!\diag{I_{(L+1)q},I_{L+1}\otimes C_{p},I_{Lm}}, \\
    \sfC_{2}&\!=\!
    \begin{bmatrix}
        O_{\ell\times((L+1)q+Ln)} & C_{p} & O_{\ell\times Lm}
    \end{bmatrix}\!, \\
    \sfE&\!=\!
    \begin{bmatrix}
        O_{(L-1)q\times q}^{\top} & I_{q} & O_{(q+(L+1)n+Lm)\times q}^{\top}
    \end{bmatrix}^{\top}\!, \\
    \sfF&\!=\!
    \begin{bmatrix}
        O_{Lq\times q}^{\top} & I_{q} & O_{((L+1)\ell+Lm)\times q}^{\top}
    \end{bmatrix}^{\top}\!.
\end{align*}
The following lemma shows that the stability of the original closed-loop system is invariant even if its plant is rewritten and the IOHFC is utilized.

\begin{lemma}\label{lem:stability}
    Suppose $r_{t}=0$ for all $t\in\N$ and $z_{t}=0$ for all $t\le 0$.
    The closed-loop system with \eqref{eq:plant_data} and \eqref{eq:controller_data} is stable if and only if that with \eqref{eq:plant} and \eqref{eq:controller} is stable.
\end{lemma}

\begin{proof}
    See \appref{app:stability}.
\end{proof}

Next, the stability condition for $\Delta_{K}$ is derived as follows.

\begin{theorem}\label{thm:lyapunov}
    Given the controller \eqref{eq:controller} stabilizing the plant \eqref{eq:plant} with $r_{t}=0$ for $t\in\N$ and $z_{t}=0$ for $t\le 0$.
    If the sensitivity $\Delta_{K}$ is chosen such that
    \begin{equation}
        \Delta_{K}<\beta_{1}\left(-\beta_{2}+\sqrt{\beta_{3}}\right)
        \label{eq:Delta_K}
    \end{equation}
    with
    \begin{align*}
        \beta_{1}&=2\left(\sqrt{(L+1)(q+\ell)m+Lm^{2}}\|\sfB^{\top}P\sfB\|\|\sfC_{1}\|\right)^{-1}, \\
        \beta_{2}&=\|(\sfA+\sfB K\sfC_{1})^{\top}P\sfB\|, \\
        \beta_{3}&=\|(\sfA+\sfB K\sfC_{1})^{\top}P\sfB\|^{2}+\lambda_{\min}(Q)\|\sfB^{\top}P\sfB\|,
    \end{align*}
    then the closed-loop system with the system \eqref{eq:plant_data} and the controller
    \begin{equation}
        u_{t}=\bar{K}d_{t}
        \label{eq:qtz_controller_gain}
    \end{equation}
    is stable, where $P$ and $Q$ are positive definite matrices satisfying $(\sfA+\sfB K\sfC_{1})^{\top}P(\sfA+\sfB K\sfC_{1})-P+Q=O$.
\end{theorem}

\begin{proof}
    See \appref{app:lyapunov}.
\end{proof}

It should be noted that if $\Delta_{K}$ satisfies the condition \eqref{eq:Delta_K}, then the control system is bounded-input bounded-output stable regardless of the choice of $\Delta_{d}$.
This is because the closed-loop system with \eqref{eq:plant_data} and \eqref{eq:qtz_controller} is given as $\sfx_{t+1}=\sfA\sfx_{t}+\sfB\bar{K}\bar{d}_{t}+\sfE r_{t}=(\sfA+\sfB\bar{K}\sfC_{1})\sfx_{t}+(\sfE+\sfB\bar{K}\sfF)r_{t}+\sfB\bar{K}\tldd_{t}$, and $\tldd$ is bounded by \eqref{eq:bound_v}.
Meanwhile, the output trajectory of the closed-loop system differs from the original trajectory.
Moreover, a quantization error of $d$ would further degrade the control performance.
The following theorem estimates the degree of performance degradation induced by quantization of $K$ and $d$.

\begin{theorem}\label{thm:error}
    Given the initial state $x_{0}$ and $\sfx_{0}=[\bfzero_{(L+1)q}\ \bfzero_{Ln}\ x_{0}^{\top}\ \bfzero_{Lm}]^{\top}$.
    Suppose that $\sfA+\sfB K\sfC_{1}$ is stable, and $\Delta_{K}$ satisfies the condition \eqref{eq:Delta_K}.
    Let $y'$ be the output of the system \eqref{eq:plant_data} with the controller \eqref{eq:qtz_controller}.
    The supremum of the error between $y(K,\sfx_{0})$ and $y'(\bar{K},\sfx_{0})$ is bounded by
    \[
        \sup_{t>0}\left\|y_{t}(K,\sfx_{0})\!-\!y'_{t}(\bar{K},\sfx_{0})\right\|\!\le\!\theta_{1}c^{2}\tau\gamma^{\tau-1}+\cfrac{\theta_{2}c^{2}}{(1-\gamma)^{2}}+\cfrac{\theta_{3}c}{1-\gamma},
    \]
    where $\theta_{1}=\|\sfC_{2}\sfB\tilde{K}\sfC_{1}\|\|x_{0}\|$, $\theta_{2}=\|\sfC_{2}\sfB\tilde{K}\sfC_{1}\|\|\sfE+\sfB K\sfF\|B_{r}$, $\theta_{3}=\|\sfC_{2}\sfB\|(\|\tilde{K}\sfF\|B_{r}+\|\bar{K}\|\eta_{d})$, and $B_{r}=\sup_{t>0}\|r_{t}\|$.
    The parameters $\gamma$, $c$, and $\tau$ are determined by
    \begin{align*}
        &\hspace{-.5eM}\max\{\rho(\sfA+\sfB K\sfC_{1}),\rho(\sfA+\sfB\bar{K}\sfC_{1})\}<\gamma<1, \\
        c&\!=\!\max_{1\le k\le M}\{1,\gamma^{-k}\|(\sfA\!+\!\sfB K\sfC_{1})^{k}\|,\gamma^{-k}\|(\sfA\!+\!\sfB\bar{K}\sfC_{1})^{k}\|\}, \\
        \tau&\!=\!\round{-(\log\gamma)^{-1}},
    \end{align*}
    where $M$ is a nonnegative integer such that $\|(\sfA+\sfB K\sfC_{1})^{k}\|<\gamma^{k}$ and $\|(\sfA+\sfB\bar{K}\sfC_{1})^{k}\|<\gamma^{k}$ for all $k\ge M$.
\end{theorem}

\begin{proof}
    See \appref{app:error}.
\end{proof}

The theorem estimates the worst-case perturbation of the output trajectory due to the encryption.
It should be noted that the upper bound can be reduced by decreasing the sensitivities because $\theta_{1}$, $\theta_{2}$, and $\theta_{3}$ decrease as $\Delta_{K}$ and $\Delta_{d}$ decrease.
Moreover, the smaller $\gamma$ is, the smaller the upper bound becomes.
This implies that the error caused by the encryption can be decreased by making the closed-loop system more stable.
Note that it is difficult to cancel the quantization errors completely because the sensitivities cannot become zero.
This is an open problem, and thus we need further development of encoding in encrypted control.

\section{Numerical Simulation}\label{sec:simulation}

This section demonstrates the feasibility of the proposed scheme using the quadruple-tank process in~\cite{Johansson00} modified to add process and measurement noises.
The model of the form \eqref{eq:plant} of the process, which is linearized around the points $h_{1}^{0}=12.4$~cm, $h_{2}^{0}=12.7$~cm, $h_{3}^{0}=1.8$~cm, $h_{4}^{0}=1.4$~cm, $v_{1}^{0}=3$~V, $v_{2}^{0}=3$~V, $\gamma_{1}=0.7$, $\gamma_{2}=0.6$ and discretized with the sampling period of $1$~s, is given as
\begin{align*}
    A_{p}&=
    \begin{bmatrix}
        0.9842 & 0      & 0.0407 & 0      \\
        0      & 0.9890 & 0      & 0.0326 \\
        0      & 0      & 0.9590 & 0      \\
        0      & 0      & 0      & 0.9672
    \end{bmatrix}, \\
    B_{p}&= 
    \begin{bmatrix}
        0.0826 & 0.0010 \\
        0.0005 & 0.0625 \\
        0      & 0.0469 \\
        0.0307 & 0
    \end{bmatrix},\ 
    C_{p}=
    \begin{bmatrix}
        0.5 & 0   & 0 & 0 \\
        0   & 0.5 & 0 & 0
    \end{bmatrix},
\end{align*}
where $x_{i}=h_{i}-h_{i}^{0}$, $u_{i}=v_{i}-v_{i}^{0}$, $h_{i}$ is a water level of the tank~$i$, $v_{i}$ is a voltage applied to the pump~$i$, and the model parameters are as follows.
The cross sections of the tanks are $A_{1}=A_{3}=28$~cm$^{2}$ and $A_{2}=A_{4}=32$~cm$^{2}$.
The cross sections of the outlet holes are $a_{1}=a_{3}=0.071$~cm$^{2}$ and $a_{2}=a_{4}=0.057$~cm$^{2}$.
The output gain is $k_{c}=0.5$~V/cm, and the input gains are $k_{1}=3.33$~cm$^{3}$/Vs and $k_{2}=3.35$~cm$^{3}$/Vs.
The gravitational acceleration is $981$~cm/s$^{2}$.
The process noise $w$ and measurement noise $v$ follow the Gaussian distribution with mean zero and variance $10^{-3}$.

The parameters in \eqref{eq:controller} of the decentralized PI controller~\cite{Johansson00} used to control the process are given as
\begin{align*}
    A&=
    \begin{bmatrix}
        1 & 0 \\
        0 & 1
    \end{bmatrix},\ 
    B=
    \begin{bmatrix}
        -1 &  0 \\
         0 & -1
    \end{bmatrix},\ 
    C=
    \begin{bmatrix}
        0.1 & 0 \\
        0   & 0.0675
    \end{bmatrix}, \\
    D&=
    \begin{bmatrix}
        -3.0 &  0 \\
         0   & -2.7
    \end{bmatrix},\ 
    E=
    \begin{bmatrix}
        1 & 0 \\
        0 & 1
    \end{bmatrix},\ 
    F=
    \begin{bmatrix}
        3.0 & 0 \\
        0   & 2.7
    \end{bmatrix},
\end{align*}
where the proportional gains are $K_{1}=3.0$ and $K_{2}=2.7$, the integral times are $T_{i1}=30$ and $T_{i2}=40$, and the controller is also discretized with the sampling period.
The gain of the corresponding IOHFC is obtained as
\begin{align*}
    K=
    &\left[\begin{matrix}
        -1.45 \!&\!  0      \!&\! -1.4 \!&\!  0      \!&\! 3.0 \!&\! 0   \!&\! 1.45 \\
         0    \!&\! -1.3163 \!&\!  0   \!&\! -1.2825 \!&\! 0   \!&\! 2.7 \!&\! 0
    \end{matrix}\right. \\
    &\left.\begin{matrix}
        0      \!&\! 1.4 \!&\! 0      \!&\! -3.0 \!&\!  0   \!&\! 0.5 \!&\! 0   \!&\! 0.5 \!&\! 0   \\
        1.3163 \!&\! 0   \!&\! 1.2825 \!&\!  0   \!&\! -2.7 \!&\! 0   \!&\! 0.5 \!&\! 0   \!&\! 0.5
    \end{matrix}\right],
\end{align*}
where the data length is $L=2$.

The parameters of the BFV encryption are chosen according to the recommendation of Homomorphic Encryption Standardization\footnote{\url{https://homomorphicencryption.org/standard/}} to satisfy $\lambda=128$~bit security; specifically, the degree of the polynomical ring is $N=4096$, the plaintext modulus $T$ is a $25$~bit prime, and the ciphertext modulus $Q$ is a $109$~bit integer.
Additionally, the right-hand side of \eqref{eq:Delta_K} is calculated as $5.0740\times 10^{-4}$, and thus we choose $\Delta_{K}=2\times10^{-4}$ and $\Delta_{d}=1\times10^{-3}$.
From \thmref{thm:error}, the worst-case output error caused by the encryption with $x_{0}=[1\ 1\ 1\ 1]^{\top}$ and $B_{r}=0.7071$ is bounded by $9.7985$, where $\gamma=0.9797$, $c=16.2783$, and $\tau=49$.

\figref{fig:performance} depicts the results of the unencrypted and encrypted decentralized PI controls with the IOHFC representation.
The dashed black lines are reference inputs.
The blue and red lines are outputs of unencrypted and encrypted controls, respectively.
The initial state of the process is $x_{0}=[1\ 1\ 1\ 1]^{\top}$.
The reference inputs are set to $[0\ 0]^{\top}$ from $0$~s to $600$~s and switched between $[0.5\ 0.5]^{\top}$ and $[-0.5\ -0.5]^{\top}$ every $200$~s from $600$~s to $1400$~s.
It should be noted that the corresponding water levels $(h_{1},h_{2})$ are $(12.4,12.7)$, $(13.4,13.7)$, and $(11.4,11.7)$.
The results show that the encrypted control inherits the stability of unencrypted control, and the outputs of encrypted control as well as those of unencrypted control track the reference inputs.
In addition, \figref{fig:error} shows the $\ell_{2}$ norm of the output error between the unencrypted and encrypted controls with $v_{t}=w_{t}=0$.
The same reference inputs as \figref{fig:performance} are used in this simulation.
The maximum error of this result is $0.007$~cm, and the result demonstrates that the performance degradation due to controller encrypion is sufficiently small.

Finally, the computation times of the BFV encryption are shown in~\tabref{tab:time}.
The minimum, average, and maximum times were calculated with $100000$ times of measurements.
All the experiments are conducted using MacBook Pro (macOS Monterey, $2.3$~GHz quad-core Intel Core i7, $32$~GB $3733$~MHz LPDDR4X).
It should be noted that the computation times of $\Mult$ include those of $\Relin$.
From \tabref{tab:num_alg} and \tabref{tab:time}, the total average times in each time step for the operator, plant, and controller are about $1.28$, $2.87$, $16.45$~ms, respectively.
Thus, the total computation time is within the sampling time.
This result suggests that the proposed method can be applied to practical real-time systems.

\begin{figure}[!t]
    \centering
    \subfigure[Water level in tank~$1$.]{\includegraphics[scale=.95]{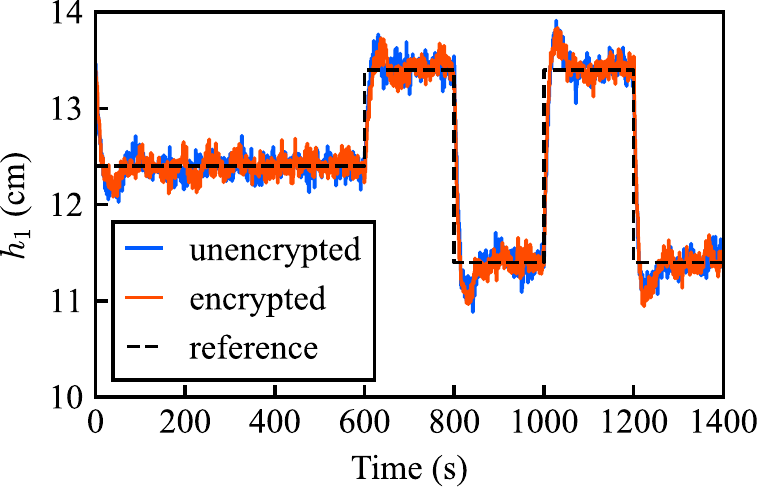}\label{fig:h1}}
    \subfigure[Water level in tank~$2$.]{\includegraphics[scale=.95]{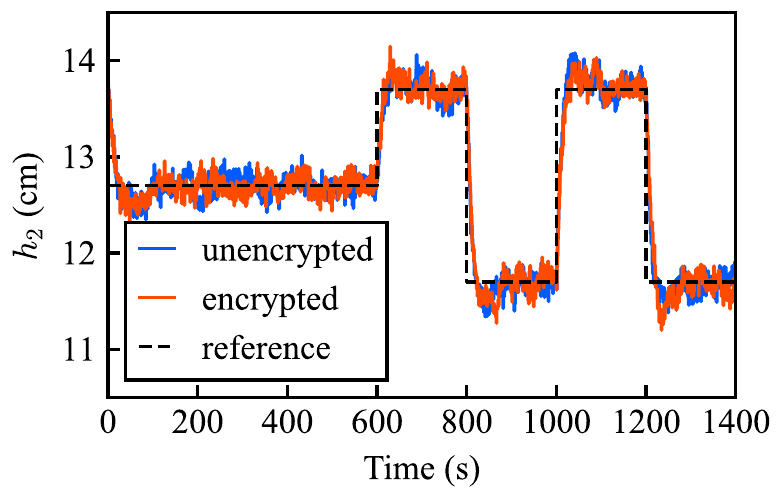}\label{fig:h2}}
    \caption{Comparison between the unencrypted and encrypted controls with the IOHFC.}
    \label{fig:performance}
\end{figure}

\begin{figure}[!t]
    \centering
    \includegraphics[scale=.95]{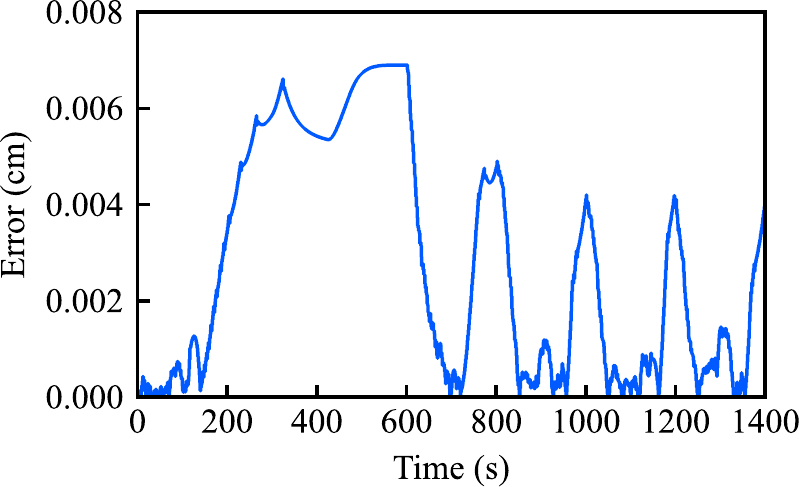}
    \caption{Performance degradation due to encryption.}
    \label{fig:error}
    \vspace{-5mm}
\end{figure}

\begin{table}[!t]
    \centering
    \caption{Computation Times}
    \begin{tabular}{cccccccc}
        \toprule
                     & $\sigma$ & $\sigma^{-1}$ & $\Enc$ & $\Dec$ & $\Mult$ & $\Add$ & $\Rotate$ \\
        \midrule
        Min (ms)     & $0.03$   & $0.04$        & $1.17$ & $0.26$ & $4.09$  & $0.01$ & $0.63$   \\
        Ave (ms)     & $0.03$   & $0.04$        & $1.25$ & $0.27$ & $4.35$  & $0.01$ & $0.66$   \\
        Max (ms)     & $0.38$   & $0.44$        & $4.99$ & $0.78$ & $10.4$  & $0.09$ & $1.41$   \\
        Std ($\mu$s) & $5.47$   & $6.99$        & $72.6$ & $22.7$ & $208$   & $2.40$ & $42.9$   \\
        \bottomrule
    \end{tabular}
    \label{tab:time}
    \vspace{-5mm}
\end{table}

\section{Conclusion}\label{sec:conclusion}

This paper presented a novel controller representation based on input-output history data of a dynamic controller to implement an encrypted controller with leveled fully homomorphic encryption.
The proposed encrypted control scheme and algorithm enable the controller to operate for an infinite time horizon without temporary decryption of the encrypted controller states.
The BFV homomorphic encryption scheme with the CRT batching improves the efficiency of matrix-vector multiplication, which is included in the proposed algorithm.

We also estimated the worst-case performance degradation caused by the quantization effects due to encryption.
The numerical simulation demonstrates the feasibility of the proposed encrypted control with a small performance degradation by choosing the appropriate parameters.
Furthermore, the simulation results disclose that the derived theoretical estimate is slightly conservative, and so we will further analyze the effects of encryption on control performance.

Future work includes the consideration of an IOHFC representation for more complex controllers, such as nonlinear controllers.
One possible way to realize a nonlinear IOHFC is use of the Koopman operator~\cite{Koopman31}, which lifts a finite-dimensional nonlinear system to an infinite-dimensional linear system.
The approach with appropriate truncation of a system dimension would enable application of the proposed scheme to nonlinear controllers.

\appendix

\subsection{BFV encryption scheme}\label{app:BFV}

The RLWE problem is defined as follows~\cite{Lyubashevsky10}.

\begin{definition}[RLWE]
    Given a security parameter $\lambda$.
    Let $Q=Q(\lambda)\ge2$ be an integer, $M=M(\lambda)$ be a power of $2$, and $\chi=\chi(\lambda)$ be a distribution over $R=\Z[X]/(\Phi_{M(\lambda)}(X))$, where $\Phi_{M}(X)$ is the $M$th cyclotomic polynomial.
    Sample $\bfs\gets R_{Q}$ randomly, and define the distribution $D(\bfs,Q,\chi)$ obtained by outputting $([\bfa\bfs+\bfe]_{Q},\bfa)$ with random sampling $\bfa\gets R_{Q}$ and $\bfe\gets\chi$.
    The RLWE problem is to distinguish between $D(Q,\bfs,\chi)$ and the uniform distribution over $R_{Q}^{2}$.
    The RLWE assumption is the assumption that the distributions are computationally indistinguishable.
\end{definition}

The BFV leveled fully homomorphic encryption~\cite{Fan12} is constructed based on the RLWE assumption as follows:

\begin{itemize}
    \item $\KeyGen(\lambda)$:
    Choose $N(\lambda)$, $Q(\lambda)$, $T(\lambda)$, $W(\lambda)$, and $\chi(\lambda)$.
    Randomly sample $\bfs\gets R_{2}$, $\bfa\gets R_{Q}$, $\bfe\gets\chi$, $\bfa'_{i}\gets R_{Q}$ and $\bfe'_{i}\gets\chi$ for $0\le i\le\ell=\floor{\log_{W}(Q)}$.
    Set
    \begin{align*}
        \sk&=\bfs,\quad\pk=([-(\bfa\bfs+\bfe)]_{Q},\bfa), \\
        \rlk&=\left[([-(\bfa'_{i}\bfs+\bfe'_{i})+W^{i}\bfs^{2}]_{Q},\bfa'_{i})\mid 0\le i\le\ell\right].
    \end{align*}
    Output $(\sk,\pk,\rlk)$.
    
    \item $\Enc(\pk,\bfm)$:
    A plaintext space is $R_{T}$.
    Let $\Delta=\floor{Q/T}$, $\bfp_{0}=\pk[0]$, and $\bfp_{1}=\pk[1]$.
    Randomly sample $\bfu\gets R_{2}$, and $\bfe_{0},\bfe_{1}\gets\chi$.
    Output
    \[
        \ct=([\bfp_{0}\bfu+\bfe_{0}+\Delta\bfm]_{Q},[\bfp_{1}\bfu+\bfe_{1}]_{Q}).
    \]
    
    \item $\Dec(\sk,\ct)$:
    A ciphertext space is $R_{Q}^{2}$.
    Let $\bfc_{0}=\ct[0]$, $\bfc_{1}=\ct[1]$, and $\bfs=\sk$.
    Output
    \[
        \bfm=[\round{(T/Q)[\bfc_{0}+\bfc_{1}\bfs]_{Q}}]_{T}.
    \]
    
    \item $\Add(\ct_{1},\ct_{2})$:
    Let $\bfc_{10}=\ct_{1}[0]$, $\bfc_{11}=\ct_{1}[1]$, $\bfc_{20}=\ct_{2}[0]$, and $\bfc_{21}=\ct_{2}[1]$.
    Output
    \[
        \ct_{\Add}=([\bfc_{10}+\bfc_{20}]_{Q},[\bfc_{11}+\bfc_{21}]_{Q}).
    \]
    
    \item $\Mult(\ct_{1},\ct_{2})$:
    Let $\bfc_{10}=\ct_{1}[0]$, $\bfc_{11}=\ct_{1}[1]$, $\bfc_{20}=\ct_{2}[0]$, and $\bfc_{21}=\ct_{2}[1]$.
    Compute
    \begin{align*}
        \bfc_{0}&=[\round{(T/Q)(\bfc_{10}\bfc_{20})}]_{Q}, \\
        \bfc_{1}&=[\round{(T/Q)(\bfc_{10}\bfc_{21}+\bfc_{11}\bfc_{20})}]_{Q}, \\
        \bfc_{2}&=[\round{(T/Q)(\bfc_{11}\bfc_{21})}]_{Q}.
    \end{align*}
    Output $\ct_{\Mult}=(\bfc_{0},\bfc_{1},\bfc_{2})$.
    
    \item $\Relin(\rlk,\ct_{\Mult})$:
    Let $\bfc_{0}=\ct_{\Mult}[0]$, $\bfc_{1}=\ct_{\Mult}[1]$, and $\bfc_{2}=\ct_{\Mult}[2]$.
    Let $\bfr_{i0}=\rlk[i][0]$ and $\bfr_{i1}=\rlk[i][1]$ for $0\le i\le\ell$.
    Write $\bfc_{2}=\sum_{i=0}^{\ell}\bfc_{2}^{(i)}W^{i}$ with $\bfc_{2}^{(i)}\in R_{W}$, where $W$ is totally independent of $T$.
    Output
    \[
        \ct=\left(\left[\bfc_{0}+\sum_{i=0}^{\ell}\bfr_{i0}\bfc_{2}^{(i)}\right]_{Q},\left[\bfc_{1}+\sum_{i=0}^{\ell}\bfr_{i1}\bfc_{2}^{(i)}\right]_{Q}\right).
    \]
\end{itemize}

\subsection{Proof of \lemref{lem:stability}}\label{app:stability}

\begin{proof}
    We prove only the sufficient condition because the proof for the necessary condition is trivial.
    Given the controller \eqref{eq:controller} that stabilizes \eqref{eq:plant}, then $x_{t}\to 0$ and $u_{t}\to 0$ as $t\to\infty$.
    \thmref{theorem:iohfc} implies that the IOHFC \eqref{eq:controller_data} is equivalent to \eqref{eq:controller} for all $t\ge 0$.
    Thus, the sequences $\{x_{t}\}_{t=0}^{\infty}$ and $\{u_{t}\}_{t=0}^{\infty}$ generated by \eqref{eq:plant_data} with \eqref{eq:controller_data} are the same as those in \eqref{eq:plant} with \eqref{eq:controller}.
    This concludes state $\sfx$ of \eqref{eq:plant_data} with \eqref{eq:controller_data} converges to zero as $t\to\infty$.
\end{proof}

\subsection{Proof of \thmref{thm:lyapunov}}\label{app:lyapunov}

\begin{proof}
    The closed-loop system with \eqref{eq:plant_data} and \eqref{eq:qtz_controller_gain} is expressed as follows:
    \[
        \sfx_{t+1}=\sfA\sfx_{t}+\sfB\bar{K}\sfC_{1}\sfx_{t}=(\sfA+\sfB K\sfC_{1})\sfx_{t}+\sfB\tldK\sfC_{1}\sfx_{t},
    \]
    where $\tldK=\bar{K}-K$.
    From \lemref{lem:stability}, there always exist positive definite matrices $P$ and $Q$ such that $(\sfA+\sfB K\sfC_{1})^{\top}P(\sfA+\sfB K\sfC_{1})-P+Q=O$.
    Let $V(\sfx_{t})=\sfx_{t}^{\top}P\sfx_{t}$ be a Lyapunov function candidate, then
    \begin{align*}
        &V(\sfx_{t+1})-V(\sfx_{t}) \\
        &=\sfx^{\top}(\sfB\tldK\sfC_{1})^{\top}P\sfB\tldK\sfC_{1}\sfx+2\sfx^{\top}(\sfA+\sfB K\sfC_{1})^{\top}P\sfB K\sfC_{1}\sfx \\
        &\quad-\sfx^{\top}Q\sfx, \\
        &\!\le\!\Big(\|\sfB^{\!\top}P\sfB\|\|\sfC_{1}\|^{2}\|\tldK\|^{2}\!+\!2\|(\sfA\!+\!\sfB K\sfC_{1})^{\!\top}P\sfB\|\|\sfC_{1}\|\|\tldK\| \\
        &\quad\quad-\lambda_{\min}(Q)\Big)\|\sfx\|^{2}\eqqcolon g(\|\tldK\|).
    \end{align*}
    The solution to the quadratic equation $g(\|\tldK\|)=0$ is
    \begin{align*}
        \|\tldK\|
        &=\cfrac{1}{\|\sfB^{\top}P\sfB\|\|\sfC_{1}\|}\Bigg(-\|(\sfA+\sfB K\sfC_{1})^{\top}P\sfB\| \\
        &\quad+\sqrt{\|(\sfA+\sfB K\sfC_{1})^{\top}P\sfB\|^{2}+\lambda_{\min}(Q)\|\sfB^{\top}P\sfB\|}\Bigg).
    \end{align*}
    Moreover, it follows from \eqref{eq:bound_K} that
    \[
        \cfrac{2}{\sqrt{(L+1)(q+\ell)m+Lm^{2}}}\|\tldK\|\le\Delta_{K}.
    \]
    Therefore, $g(\|\tldK\|)<0$ if $\Delta_{K}$ satisfies \eqref{eq:Delta_K}.
    This implies that $V(\sfx_{t+1})-V(\sfx_{t})$ is negative.
\end{proof}

\subsection{Proof of \thmref{thm:error}}\label{app:error}

\begin{proof}
    Let $\Acl=\sfA+\sfB K\sfC_{1}$ and $\Aclb=\sfA+\sfB\bar{K}\sfC_{1}$.
    Because $\Acl$ and $\Aclb$ are assumed to be stable, there exists $M\ge0$ such that $\|\Acl^{k}\|<\gamma^{k}$ and $\|\Aclb^{k}\|<\gamma^{k}$ for all $k\ge M$~\cite{Dowler13}.
    Thus, $\|\Acl^{k}\|\le c\gamma^{k}$ and $\|\Aclb^{k}\|\le c\gamma^{k}$ for any $k$.
    
    It follows from \eqref{eq:plant_data}, \eqref{eq:controller_data}, and \eqref{eq:qtz_controller} that
    \begin{align*}
        y_{t}(K,\sfx_{0})&=\sfC_{2}\Acl^{t}\sfx_{0}+\sum_{k=0}^{t-1}\sfC_{2}\Acl^{k}(\sfE+\sfB K\sfF)r_{t-1-k}, \\
        y'_{t}(\bar{K},\sfx_{0})&=\sfC_{2}\Aclb^{t}\sfx_{0}+\sum_{k=0}^{t-1}\sfC_{2}\Aclb^{k}(\sfE+\sfB \bar{K}\sfF)r_{t-1-k} \\
        &\quad+\sum_{k=0}^{t-1}\sfC_{2}\Aclb^{k}\sfB\bar{K}\tldd_{t-1-k}.
    \end{align*}
    Hence, the supremum of the error is bounded by
    \begin{align*}
        &\sup_{t>0}\|y_{t}(K,\sfx_{0})-y'_{t}(\bar{K},\sfx_{0})\| \\
        &\le\sup_{t>0}\left\|\sum_{k=0}^{t-1}\Acl^{t-1-k}\Aclb^{k}\right\|\theta_{1}+\sup_{t>0}\left\|\sum_{k=1}^{t-1}\sum_{j=0}^{k-1}\Acl^{k-1-j}\Aclb^{j}\right\|\theta_{2} \\
        &\quad+\sup_{t>0}\left\|\sum_{k=0}^{t-1}\Aclb^{k}\right\|\theta_{3}, \\
        &\le\sup_{t>0}\theta_{1}c^{2}t\gamma^{t-1}+\sup_{t>0}\theta_{2}c^{2}\sum_{k=1}^{t-1}k\gamma^{k-1}+\sup_{t>0}\theta_{3}c\sum_{k=0}^{t-1}\gamma^{k}.
    \end{align*}
    For the first term of the above inequality, it follows that
    \[
        \cfrac{\partial}{\partial t}\,t\gamma^{t-1}=\gamma^{t-1}(1+t\log\gamma)=0\iff t=-(\log\gamma)^{-1},
    \]
    where $t\neq 0$.
    Furthermore, the second and third terms are respectively calculated as
    \begin{align*}
        \sup_{t>0}\theta_{2}c^{2}\sum_{k=1}^{t-1}k\gamma^{k-1}&=\theta_{2}c^{2}\sum_{k=0}^{\infty}k\gamma^{k-1}=\cfrac{\theta_{2}c^{2}}{(1-\gamma)^{2}}, \\
        \sup_{t>0}\theta_{3}c\sum_{k=0}^{t-1}\gamma^{k}&=\theta_{3}c\sum_{k=0}^{\infty}\gamma^{k}=\cfrac{\theta_{3}c}{1-\gamma},
    \end{align*}
    as $\gamma<1$.
    Therefore, we obtain
    \begin{align*}
        &\sup_{t>0}\theta_{1}c^{2}t\gamma^{t-1}+\sup_{t>0}\theta_{2}c^{2}\sum_{k=1}^{t-1}k\gamma^{k-1}+\sup_{t>0}\theta_{3}c\sum_{k=0}^{t-1}\gamma^{k} \\
        &\le\theta_{1}c^{2}\tau\gamma^{\tau-1}+\cfrac{\theta_{2}c^{2}}{(1-\gamma)^{2}}+\cfrac{\theta_{3}c}{1-\gamma}.
    \end{align*}
    This completes the proof.
\end{proof}

\bibliographystyle{IEEEtran}
\bibliography{encrypted_control_and_optimization}

\begin{IEEEbiography}[{\includegraphics[width=1in,height=1.25in,clip,keepaspectratio]{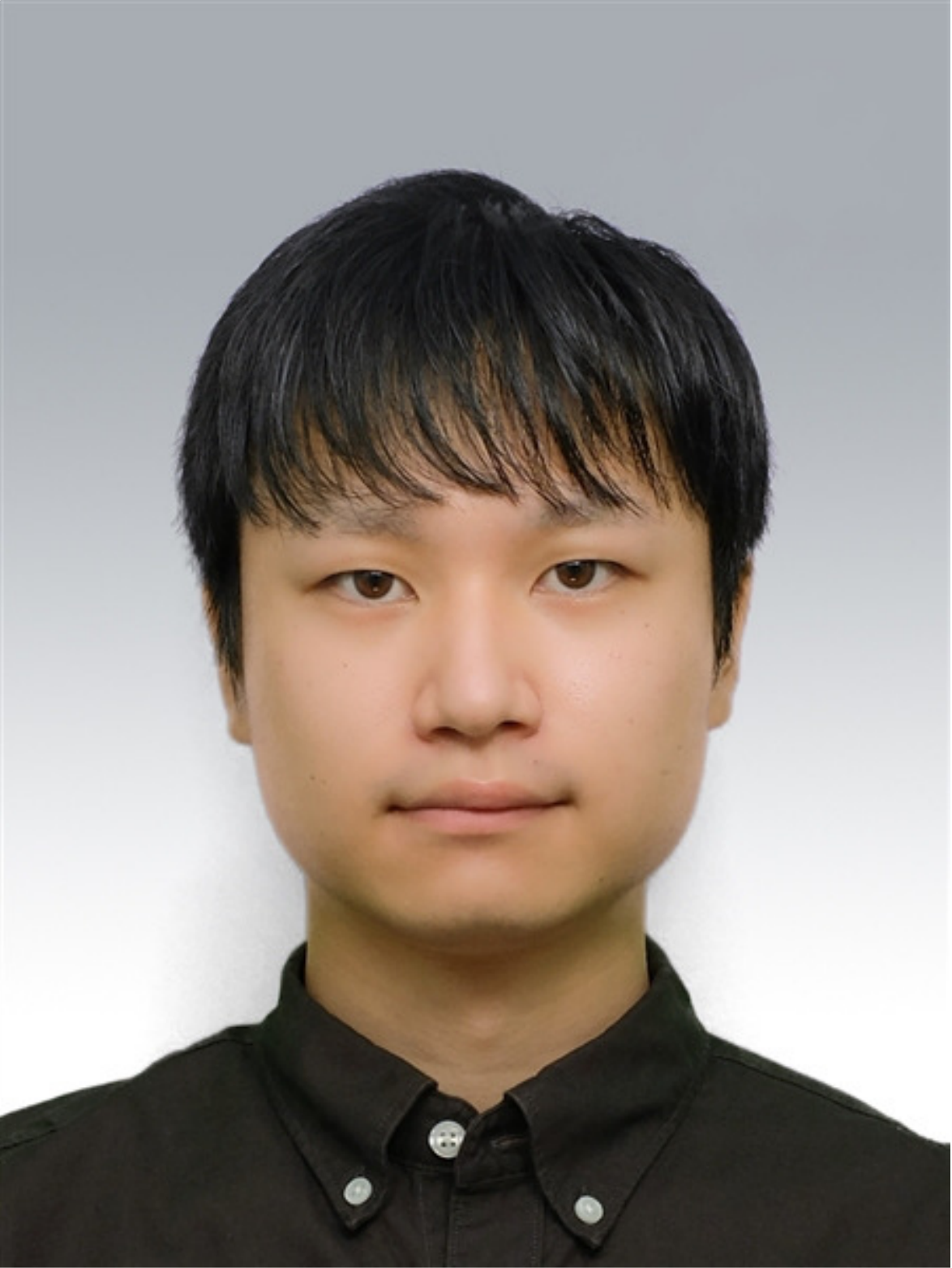}}]
    {Kaoru Teranishi} received the B.S. degree in electromechanical engineering from National Institute of Technology, Ishikawa College, Ishikawa, Japan, in 2019.
    He also obtained the M.S. degree in Mechanical and Intelligent Systems Engineering from The University of Electro-Communications, Tokyo, Japan, in 2021.
    He is currently a Ph.D. student at The University of Electro-Communications.
    From October 2019 to September 2020, he was a visiting scholar of the Georgia Institute of Technology, GA, USA.
    Since April 2021, he has been a Research Fellow of Japan Society for the Promotion of Science.
    His research interests include control theory and cryptography for cyber-security of control systems.
\end{IEEEbiography}

\begin{IEEEbiography}[{\includegraphics[width=1in,height=1.25in,clip,keepaspectratio]{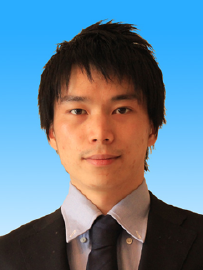}}]
    {Tomonori Sadamoto} received the Ph.D. degree from the Tokyo Institute of Technology, Tokyo, Japan in 2015. From June in 2015 to March in 2016, he was a Visiting Researcher at School of Electrical Engineering, Royal Institute of Technology, Stockholm, Sweden. From April 2016 to August 2016, he was a researcher with the Department of Systems and Control Engineering, Graduate School of Engineering, Tokyo Institute of Technology. From August 2016 to November 2018, he was a specially appointed Assistant Professor with the same department. Since November 2018, he has been assistant professor with Department of Mechanical and Intelligent Systems Engineering in the University of Electro-Communications. He was named as a finalist of the European Control Conference Best Student-Paper Award in 2014. He received Research encouragement award from The Funai Foundation for Informaiton Technology in 2019, and received IEEE Control Systems Magazine Outstanding Paper Award in 2020. 
\end{IEEEbiography}

\begin{IEEEbiography}[{\includegraphics[width=1in,height=1.25in,clip,keepaspectratio]{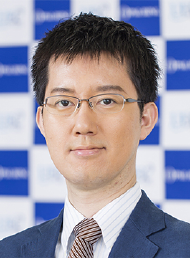}}]
    {Kiminao Kogiso} received the B.E., M.E., and Ph.D. degrees in mechanical engineering from Osaka University, Japan, in 1999, 2001, and 2004, respectively. 
    He was appointed as a postdoctoral fellow in the 21st Century COE Program and as an Assistant Professor in the Graduate School of Information Science, Nara Institute of Science and Technology, Nara, Japan, in April 2004 and July 2005, respectively.
    From November 2010 to December 2011, he was a visiting scholar at Georgia Institute of Technology, Atlanta, GA, USA.
    In March 2014, he was promoted to the position of Associate Professor in the Department of Mechanical and Intelligent Systems Engineering at The University of Electro-Communications, Tokyo, Japan. 
    And since April 2023, he has been serving as a full Professor in the same department.
    His research interests include cybersecurity of control systems, constrained control, control of decision-makers, and their applications.
\end{IEEEbiography}

\end{document}